\documentclass[]{IEEEtran}

\usepackage{algorithm}
\usepackage{algorithmic}
\usepackage{color}

\usepackage{epsfig}
\usepackage{graphics} 
\usepackage{graphicx}
\graphicspath{{./Figures/}}
\usepackage{amssymb}  
\usepackage{dsfont}
\usepackage[english]{babel}
\usepackage{amsmath} 
\usepackage{tikz} %
\usepackage{psfrag}
\usepackage{dsfont}

\newtheorem{teo}{Theorem}

\newtheorem{proposition}[teo]{Proposition}

\newtheorem{lemma}[teo]{Lemma}
\newtheorem{remark}{Remark}

\usepackage{colortbl}
    \definecolor{gray1}{rgb}{0,0,.8}
    \definecolor{gray2}{rgb}{0,1,0}
    \definecolor{gray3}{rgb}{0.8,0,.2}

\newcommand{\labeltext}[1]
{   \ifcase#1
    \or 24
    \or 20
    \or 5
    \or 2
    \or 10
    \or 6
    \fi
}

\newcommand{\argmax}[1]{\underset{#1}{\mathrm{argmax\,}}}

\newcommand{\argmin}[1]{\underset{#1}{\mathrm{argmin\,}}}

\newtheorem{es}{Example}

\newcommand{\map}[3]{#1: #2 \rightarrow #3}
\newcommand{\until}[1]{\{1,\ldots,#1\}}

\newcommand{\realpositive}{\mathbb{R}_{>0}}
\newcommand{\realnonnegative}{\mathbb{R}_{\geq0}}

\newcommand{\integernonnegative}{\mathbb{Z}_{\ge 0}}

\newcommand{\R}{\mathbb{R}} 
\newcommand{\N}{\mathbb{N}}  

\newcommand{\G}{\mathcal{G}} 
\newcommand{\V}{\mathcal{V}} 
\newcommand{\E}{\mathcal{E}} 

\newcommand{\eps}{\varepsilon}

\newcommand{\e}{\mathrm{e}} 

\newcommand{\1}{\mathbf{1}} 
\newcommand{\ind}{\mathds{1}} 

\usepackage{url}
\usepackage{soul}	
\setul{2pt}{2pt} 		
\usepackage{bm}
\usepackage[
Symbol						
]{upgreek}
\usepackage{tabularx}

\newcommand{\lessnew}[1]{{\color{black}#1}}

\title{\lessnew{Distributed estimation from relative measurements of heterogeneous and uncertain quality}%
\thanks{The research leading to this paper has been partly developed while C.~Ravazzi was at Department of Electronics and Telecommunications (DET), Politecnico di Torino, Italy, and  N.~P.~K. Chan and P. Frasca were with the Department of Applied Mathematics, University of Twente, Enschede, the Netherlands. The authors' research has been partially supported by the International Bilateral Joint CNR Lab COOPS and by IDEX Universit\'e Grenoble Alpes under C2S2 ``Strategic Research Initiative'' grant.}%
}%
\author{Chiara~Ravazzi\thanks{Chiara Ravazzi is with the National Research Council (CNR), Institute of Electronics, Computers and Telecommunication Engineering (IEIIT), c/o Politecnico di Torino, Italy. E-mail: {chiara.ravazzi@ieiit.cnr.it}.}, \and Nelson~P.~K. Chan\thanks{Nelson~P.~K. Chan is with the University of
Groningen, Groningen, the Netherlands. E-mail: {n.p.k.chan@rug.nl}.}, \and Paolo~Frasca\thanks{Paolo Frasca is with Univ.\ Grenoble Alpes, CNRS, Inria, Grenoble INP, GIPSA-lab, F-38000 Grenoble, France and Research Associate at the IEIIT-CNR, Torino, Italy. E-mail:paolo.frasca@gipsa-lab.fr.}}

\begin{document}

\maketitle
\thispagestyle{empty}

\begin{abstract}
\lessnew{This paper studies the problem of estimation from relative measurements in a graph, in which a vector indexed over the nodes has to be reconstructed from pairwise measurements of differences between its components associated to nodes connected by an edge.
In order to model heterogeneity and uncertainty of the measurements, we assume them to be affected by additive noise distributed according to a Gaussian mixture. In this original setup, we formulate the problem of computing the Maximum-Likelihood (ML) estimates and we design two novel algorithms, based on Least Squares regression and Expectation-Maximization (EM). The first algorithm (LS-EM) is centralized and performs the estimation from relative measurements, the soft classification of the measurements, and the estimation of the noise parameters. The second algorithm (Distributed LS-EM) is distributed 
 and performs estimation and soft classification of the measurements, but requires the knowledge of the noise parameters. 
We provide rigorous proofs of convergence of both algorithms and we present numerical experiments to evaluate and compare their performance with classical solutions. The experiments show the robustness of the proposed methods against different kinds of noise and, for the Distributed LS-EM, against errors in the knowledge of noise parameters.}
\end{abstract}

\section{Introduction}
Whenever measurements are used to estimate a quantity of interest, measurement errors must be properly taken into account and the statistical properties of these errors  should be identified to enable an efficient estimation. In this paper, we look at a specific case of this broad issue, within the context of network systems. Namely, we consider the problem of distributed estimation from relative measurements, defined as follows.
We assume to have a real vector that is indexed over the nodes of a graph with a known topology: the nodes are allowed to take pairwise measurements of the differences between their vector entries and those of their neighbors in the graph. The estimation problem consists in reconstructing the original vector, up to an additive constant. This prototypical problem can be applied in a variety of contexts~\cite{PB-JPH:08}. One example is relative localization of mobile automated vehicles, where the vehicles have to locate themselves by using only distance measurements~\cite{PB-JPH:07}. Another example is statistical ranking, where a set of items needs to be sorted according to their quality, which can only be evaluated comparatively~\cite{JX-LL-YY-YY:11,BO-CB-SO:14}.
In all these scenarios, the noise affecting the measurements can be drastically heterogeneous and, more importantly, its distribution may not be known a priori. For instance, in vehicle localization, distances between vehicles may be measured by more or less accurate sensors; in a ranking system, the items upon evaluation can be compared by more or less trustworthy entities. It is thus important to identify unreliable measurements and weight them differently in the estimation. In order to model this uncertainty, in this paper we assume that the measurement noise is sampled from a mixture of two Gaussian distributions with different variances, representing good and poor measurements, respectively. Our solution to this problem builds on the classical Expectation-Maximization (EM) approach~\cite{AD-NL-DR:77,TKM:96}, where the likelihood is maximized by alternating operations of expectation and maximization.

We are particularly interested in finding \lessnew{efficient} \textit{distributed} algorithms to solve this problem. More precisely, we say that an algorithm is distributed if it requires each node to use information that is directly available at the node itself or from its immediate neighbors. Actually, many distributed algorithms for relative estimation are available\lessnew{~\cite{AG-PRK:06a,PB-JPH:08,SB-SDF-LS-DV:10,WSR-PF-FF:17,AC-MT-RC-LS:13,NMF-AZ:12,CR-PF-RT-HI:13c,PF-HI-CR-RT:15,7810670}, but they assume that the quality of the measurements is known beforehand.}
At the same time, there is a large literature on robust estimation that also covers estimation from relative measurements, but often provides algorithms that are not distributed; see for instance~\cite{LC-AC-FD:14} and references therein. \lessnew{Perhaps the only work on robust distributed relative estimation is the recent~\cite{BoToCaSc2016}: their approach is very different from ours as it is based on $\ell_1$ optimization. 
By proposing our distributed EM algorithm,} we contribute to the growing body of research on distributed algorithms for network-related estimation problems with heterogeneous and unknown measurements~\cite{AC-FF-LS-SZ:10,FF-SF-CR:14,FFR2, BoToCaSc2016,3804a93088e144e992c75dec01319d0f}, where distributed algorithms based on consensus and ranking procedures have been proposed to approximate Maximum-Likelihood (ML) estimates.

Other authors have used EM to estimate Gaussian mixtures' parameters in other problems of distributed inference in sensor networks~\cite{1212679,1307319,citeulike:13511870}. In these works, a network is given and each node independently performs the E-step from local observations and this information is suitably propagated to collaboratively perform the M-step.  
EM is also a key instrument to design reliable learning systems based on unreliable information reported by users in the context of social sensing~\cite{DW-LK-TA:14,DW-TA-LK:15}.

\subsubsection*{Our contribution}
In this paper, we define the problem of robust estimation from relative measurements when measurement noise is drawn from a Gaussian mixture and we design two iterative algorithms that solve it. Both algorithms are based on combining classical Weighted Least Squares (WLS) with Expectation-Maximization (EM), which is a popular tool in statistical estimation problems involving incomplete data~\cite{Bishop:2006:PRM:1162264,MRG-YC:11}. The first algorithm is centralized, whereas the second algorithm is distributed but requires to know (approximately) the two variances of the Gaussian mixture. This knowledge is not necessary for the centralized version. Both algorithms are proved to converge and their performance is compared on synthetic data. We observe that the centralized algorithm has better performance, achieving smaller estimation errors. The centralized algorithm also requires less iterations to converge, but each iteration involves more computations.

\subsubsection*{Organization of the paper}
We formally present the problem of relative estimation in Section~\ref{sect:localization-problem}, \lessnew{where we also review some state-of-the-art algorithms.}  
The centralised LS-EM algorithm is described in Section~\ref{sect:algo} and the Distributed LS-EM algorithm in Section~\ref{sect:algo-distrib}. Section~\ref{sect:sims} contains some numerical examples and Section~\ref{sect:outro} our conclusions. The details of our proofs are postponed to the Appendix.

\subsubsection*{Notation}
\lessnew{Throughout this paper, we use the following notational conventions. Real
and nonnegative integer numbers are denoted by $\R$ and $\integernonnegative$,
respectively. 
Open intervals are denoted by parentheses, while closed intervals are denoted by square brackets. 
Given a finite set $\mathcal{V}$, $\R^{\mathcal{V}}$ denotes the Eucliean space of real vectors with components labelled  by elements of $\mathcal{V}$.
We denote column vectors with small letters, and matrices with capital letters. Given $x\in\R^\mathcal{V}$, we denote its $v$-th element by $x_v$ or $(x)_v$. Given $x\in\R^{\mathcal{V}}$ and $A\in\R^{\mathcal{V}\times \mathcal{V}}$, we denote the $\ell_p$ norm of vector $x$ with the symbol $\|x\|_p$ (the $ \ell_{2} $ norm is taken when subscript $ p $ is omitted), and with $\|A\|$ the spectral norm of matrix $A$.
The support set of a vector $x$ is defined by $\mathrm{supp}(x)=\{i\in\mathcal{V}:x_i\neq0\}$ and we define $\Sigma_k=\{x\in\R^\mathcal{V}:\|x\|_0\leq k\}$ with $\|x\|_0=|\mathrm{supp}(x)|$ denoting the $\ell_0$-pseudonorm. Given $\mathcal{E}$ with finite cardinality $|\mathcal{E}|$, we define $\map{\mathcal{P}_{\ell}}{[0,1]^\mathcal{E}}{\Sigma_{|\mathcal{E}|-\ell}}$ as the projection that zeroes the $\ell$ smallest components of the given vector. It should be noticed that in general the projection of a vector could be not unique: we assume that $\mathcal{P}_{\ell}(x)$ consistently selects one of the possible projections by a tie-breaking rule. 
Given a matrix $M$, $M^\top$ denotes its transpose. Given a vector $x$, we denote with $\mathrm{diag}(x)$ the diagonal matrix whose diagonal entries are the elements of $x$.

An (undirected) graph is a pair $\mathcal{G}=(\mathcal{V, E})$ where $\mathcal{V}$ is a set, called the set of vertices, and $\mathcal E\subseteq \{\{v,w\}: v, w \in \mathcal{V}\}$ is the set of edges. 
Graph $\mathcal{G}$ is connected if, for all $i,j\in \mathcal V$, there exist vertices $i_1,\dots i_s$ such that $\{i,i_1\}, \{i_1,i_2\},\dots ,\{i_s,j\}\in \mathcal E$. We let $A\in\{0,\pm1\}^{\mathcal{E}\times\mathcal{V}}$ be the edge incidence matrix of the graph $\mathcal{G}$, defined as follows. The rows and the columns of $A$ are indexed by elements of  $\E$ and $\V$, respectively. We assume to have an order on set $\mathcal{V}$, such as it would be for $\mathcal{V}=\until{n}$. By this order, the orientation of the edges is conventionally assumed to be such that, if $u<v$, then edge $\{v,u\}$ originates in $u$ and terminates in $v$.
The $(e, w)$-entry of $A$ is 0 if vertex $w$ and edge $e$ are not incident, and otherwise it is $1$ or $-1$ according as $e$ originates or terminates at $w$:
$$
A_{ew}=\begin{cases}
+1&\text{if } e= (v,w)\\
-1&\text{if } e=(w,v)\\
0&\text{otherwise.}\\
\end{cases}
$$
}
\section{Robust estimation from relative measurements}\label{sect:localization-problem}

A set of nodes $\V=\{1,\ldots,N\}$ of cardinality $N$ is considered, each of them endowed with an unknown scalar quantity $\widetilde x_v\in\R$ with $v\in\mathcal{V}$. Starting from a set of noisy measurements, the nodes' goal is to estimate their own absolute position.
More precisely, each node $u\in \mathcal{V}$ is interested in estimating the scalar value $\widetilde x_u$, based on noisy measurements of differences $\widetilde x_u-\widetilde x_v$ with $v$ and $u$ in $\V$. 
The set of available measurements can be conveniently represented by 
graph $\mathcal{G}=(\mathcal{V},\mathcal{E})$, where each edge represents a measurement: $A\in\{0,\pm1\}^{\mathcal{E}\times\mathcal{V}}$ is the edge incidence matrix of graph $\mathcal{G}$.
We let $b\in \mathbb{R}^{\E} $ be the vector collecting the measurements
$$
b=A \widetilde x+\eta,
$$
where $\eta_e,e\in\E$ are mutually independent random variables distributed according to a Gaussian distribution $\mathcal{N}(0,\sigma^2_e)$, having
$$\sigma_{e}=(1-{z}_e)\alpha+ {z}_e\beta,$$ with $0<\alpha<\beta$, with ${z}_e$ distributed as a Bernoulli distribution ${z}_e\sim\mathcal{B}(p)$ and $p\in(0,1/2)$. 
Provided $\alpha \ll \beta$, the value ${z}_e=1$ is associated to a measurement that is unreliable. 
With this formulation the random variables $\{\eta_e\}_{e\in\mathcal{E}}$ are Gaussian mixtures, whose model is completely described by three parameters: $p, \alpha$ and $\beta$.
\lessnew{For convenience, from now on we consider $p$ fixed and known. 
This choice is done for simplicity and does not entail a significant restriction to our analysis: on the one hand, the algorithms we propose are fairly robust to small errors in the estimate of $p$; on the other hand, our framework can be easily extended to include the estimation of $p$ as an unknown parameter.}

Our main goal is to obtain a robust estimate of the state vector $\widetilde{x}$ by suitably taking into account the different quality of the measurements. We thus consider a joint Maximum Likelihood estimation
{\begin{equation}\label{eq:first-ML}
\widehat{x}^{\text{ML}}=\argmax{x\in\R^{{\V}},\alpha>0,\beta>0}L(x,\alpha,\beta)
\end{equation}
where $L(x,\alpha,\beta):=\log f(b|x,{\alpha,\beta})$ and
\begin{align}\label{eq:cond_distr}\begin{split}
 f(b|x, {\alpha,\beta})=\prod_{e\in\mathcal{E}}&\left[\frac{1-p}{\sqrt{2\uppi\alpha^2}}\exp{\left(-\frac{(b-Ax)_e^2}{2\alpha^2}\right)}\right.\\
&\left.+\frac{p}{\sqrt{2\uppi\beta^2}}\exp{\left(-\frac{(b-Ax)_e^2}{2{\beta}^2}\right)}\right].\end{split}
\end{align}}
The computational complexity of optimization problem~\eqref{eq:first-ML} makes  a brute force approach infeasible for large graphs.

\subsection{Estimation via Weighted Least Squares}\label{sect:wls}
Problem~\eqref{eq:first-ML}  becomes much simpler if we assume to know the distribution that has produced the noise term for each measurement. Using the noise source information $\alpha$, $\beta$, and $\widetilde{z}_{e}$ for all $e\in \mathcal{E}$, where $\widetilde{z}_{e}$ is the realization of $z_{e}$,  the ML-estimation becomes
\begin{equation}\label{eq:ML0}
X_{\mathrm{ML}}=\argmax{x\in\R^{{\V}}}\log f(b|x,\widetilde{z},\alpha,\beta)
\end{equation}
{\lessnew{where $X_{\mathrm{ML}}$ is the set of maximizing values of the log-likelihood}}
\begin{align*}
 f(b|x,\widetilde{z},\alpha,\beta)
&=\prod_{e\in\mathcal{E}}\left[\frac{1-\widetilde{z_e}}{\sqrt{2\uppi\alpha^2}}\exp{\left(-\frac{(b-Ax)_e^2}{2\alpha^2}\right)}\right.\\
&\left.\qquad+\frac{\widetilde{z_e}}{\sqrt{2\uppi\beta^2}}\exp{\left(-\frac{(b-Ax)_e^2}{2{\beta}^2}\right)}\right].
\end{align*}
Noticing that $\log(\prod_{e}x_e)=\sum_e\log(x_e)$ and that $\widetilde{z} $ is a binary vector, it is easy to see that ML is equivalent to solving the Weighted Least Square (WLS) problem 
\begin{equation}\label{eq:WLS}
\argmin{x\in\R^\V}\frac{1}{2}\|b-Ax\|^2_{W}=\argmin{x\in\R^\V}\frac{1}{2}(b-Ax)^{\top}W(b-Ax)
\end{equation}
with ${W=\mathrm{diag}(\left(1-\widetilde{z}_e\right)\alpha^{-2}+\widetilde{z}_e{\beta^{-2}})}.$

The following lemma describes the solutions of~\eqref{eq:WLS}.

\begin{lemma}[WLS estimator]\label{lemma:centralized-WLS} 
Let the graph $\mathcal{G}$ be connected and $X$ be the set of solutions of~\eqref{eq:WLS}, and let $L_W:=A^\top WA$ denote the weighted Laplacian of the graph.
The following facts hold:
\begin{enumerate}
\item $x \in X_{\mathrm{ML}}$ if and only if $A^\top WAx =A^\top W b$;
\item there exists a unique $\widehat{x}^{\text{wls}}\in X_{\mathrm{ML}}$ such that $\|\widehat{x}^{\text{wls}}\|_2=\min_{x\in X_{\mathrm{ML}}}\|x\|_2$;
\item \begin{equation}\label{eq:wls}\widehat{x}^{\text{wls}}=L^{\dag}_W A^\top {W} b,\end{equation}
where $L^{\dag}_W $ denotes the Moore-Penrose pseudo-inverse of the weighted Laplacian 
$L_W$.
\end{enumerate}
\end{lemma}
We recall that $\1^{\top}L_W=0$ and $\1^{\top}L^{\dag}_W=0$.
Further useful properties are collected in the following result~\cite[Sect.~5.4]{FF-PF:17}.
\begin{proposition}[Moments of WLS estimator]\label{moments}
Provided $\mathcal G$ is connected, it holds that
\begin{align*}
\mathbb{E}[\widehat{x}^{\text{wls}}]=&\big(I-\frac{1}{N}\ind \ind^{\top}\big)\widetilde{x}\\
\mathbb{E}[(\widehat{x}^{\text{wls}}-\mathbb{E}[\widehat{x}^{\text{wls}}])&(\widehat{x}^{\text{wls}}-\mathbb{E}[\widehat{x}^{\text{wls}}])^{\top}]=L_W^{\dag},
\end{align*}
where $\ind$ is a vector of length $N$ whose entries are all 1.
\end{proposition}

\smallskip
It should be stressed that determining the state vector $\widetilde{x}$ from relative measurements is only possible up to an additive
constant, being $A\ind=0$, and $X_{\mathrm{ML}}=\widehat{x}^{\text{wls}}+\text{span}(\ind)$. 
This ambiguity can be avoided by assuming the centroid of the nodes as the origin of the Cartesian coordinate system.
In view of this comment and of the results above, we shall assume from now on that $\G$ is {\em connected} and $\ind^\top\widetilde{x}=0.$

As shown in Lemma~\ref{lemma:centralized-WLS}, the WLS solution is explicitly known and
can be easily computed solving a linear system. Furthermore, the following distributed
computation is also possible, by using a gradient descent algorithm. Observe that the gradient of the cost function in~\eqref{eq:WLS} is given by $L_Wx -A^{\top}Wb$. Set an initial condition $x^{(0)} = 0$ and fix $\tau>0$ and consider
\begin{equation}\label{eq:gradient_descent}
x^{(t+1)}=(I-\tau L_W)x^{(t)} +\tau A^{\top}Wb.
\end{equation}
Provided $\tau<2\|L_W\|_2^{-1}$, the gradient descent algorithm~\eqref{eq:gradient_descent} converges to the WLS solution~\cite{WSR-PF-FF:12,WSR-PF-FF:17}.

\subsection{Relations with literature and numerical example}
The WLS estimation and the subsequent developments in this paper share some ideas with several approaches in literature. We recall two methods based on optimization that focus on {\em Sparse outliers detection} \cite{LC-AC-FD:14} and {\em Least absolute estimation} \cite{Blum:12}. Then we will summarize the main advantages of WLS in the considered setting in contrast with these methods.

The problem of finding the smallest set that contains the outliers is considered in \cite{LC-AC-FD:14}. Using the same rationale of the {\em big $M$ trick approach}, introduced in \cite[Section III.C ]{LC-AC-FD:14}, and recalling that $|y_e-(A\widetilde{x})_e|\leq 3\sigma_e$ with a probability close to 1 (about 0.997), a reasonable adaptation of \cite{LC-AC-FD:14} can be formalized as an optimization problem in the $\ell_0$-pseudonorm:
\begin{align}\begin{split}\label{eq:l0}
&\min_{x \in\R^\V, \: z\in\{0,1\}^{\mathcal{E}}:\ \mathds{1}^{\top}x=0}\|z\|_0\\
&\qquad\text{s.t. }|y_e-(Ax)_e|\leq 3\alpha+3z_e(\beta-\alpha)\quad \forall e\in\mathcal{E}.\end{split}
\end{align}
The decision variables $z_e\in\mathcal{E}$ play the role of indicator variables for each measurement $e\in\mathcal{E}$. The test to label the measurements is based on a confidence interval: if $z_e$ is 0 then the measurement is trusted, if $z_e$ is 1 then the measurement is not trusted. This problem is combinatorial and becomes intractable for large scale problems.

For this reason, a standard approach is resorting to least absolute estimation \cite{Blum:12}, also known as $\ell_1$-minimization.
Problem \eqref{eq:l0} is relaxed by replacing the $\ell_0$-pseudonorm with the $\ell_1$-norm, which is expected to promote sparsity \cite{candes2007}:
\begin{align}\begin{split}\label{eq:l1_epi}
&\min_{x \in\R^\V, \: z\in\R^{\mathcal{E}}:\ \mathds{1}^{\top}x=0}\|z\|_1\\
&\qquad\text{s.t. }|y_e-(Ax)_e|\leq 3\alpha+3z_e(\beta-\alpha)\quad \forall e\in\mathcal{E}\end{split}
\end{align}
or
\begin{equation}\label{eq:l1}
\min_{x\in\R^\V:\ \mathds{1}^{\top}x=0}\|y-Ax\|_1
\end{equation}The problem in \eqref{eq:l1} has also a probabilistic characterization and can be interpreted as ML estimation assuming that the noise is distributed according to a Laplace distribution. 
The $\ell_1$-norm is less sensitive to outliers \cite{Blum:12} and performs better than LS-estimator in presence of different types of corrupted measurements. It should be noticed that the problem in \eqref{eq:l1} is not smooth but is still convex, indeed it is a linear program (LP) and can be solved efficiently by iterative algorithms, e.g. using subgradient methods \cite{Boyd:2004:CO:993483} or iterative reweighted least squares (IRLS, \cite{citeulike:9578552}) 
that admit a distributed
implementation. Observe that the subgradient of the cost function in~\eqref{eq:l1} is given by $A^\top\mathrm{sgn}(y-Ax)$. Set an initial condition $x^{(0)} = 0$ and fix $\tau>0$ and consider
\begin{equation}\label{eq:dlae}
x^{(t+1)}=x^{(t)} +\tau A^\top\mathrm{sgn}(y-Ax).
\end{equation}
Despite these interesting features, LAE has some drawbacks. First, there are no guarantees that the solution of \eqref{eq:l1} has the minimum cardinality property. Moreover, there are no theoretical conditions under which the problem in \eqref{eq:l0} is equivalent to \eqref{eq:l1}. Extensive numerical results show that $\ell_1$-norm encourage sparsity but in general the solution of \eqref{eq:l0} and \eqref{eq:l1} do not coincide \cite{LC-AC-FD:14}.
%
%
%
Using the noise source information $\alpha, \beta$, and $\widetilde{z}_e$ for all $e\in\mathcal{E}$, problems \eqref{eq:l0} and \eqref{eq:l1_epi} reduce to a linear feasibility program. In this sense, if $\widetilde{z}_e$ is 0, then the measurement is selected, if $z_e$ is 1 then the measurement is detected as outlier and not taken into account in the search of $x\in\R^{\V}$ satisfying the constraints. WLS instead uses all the measurements in the estimation by mitigating the effect of outliers: its covariance is given in Proposition~
\ref{moments}.
Furthermore, finding the optimal estimate using WLS approach is equivalent to solving a network of resistors~\cite{PB-JPH:09}. This intuitive electrical interpretation highlights the role of the topology of the measurement graph and allows distinguishing between topologies that lead to small or large estimation errors \cite{WSR-PF-FF:13,WSR-PF-FF:17}. In particular, using Proposition~\ref{moments} and the electrical interpretation, one can relate the measurement graph $\G$ to the error in the estimation. 
Such analysis of performance is not available for $\ell_0$ or $\ell_1$-minimization.
\medskip

Finally, we provide a numerical example for illustration.
\begin{es}\label{es1}
Consider the connected network in Figure~\ref{fig:es1} with $N=5$ nodes and set of edges $\E=\{(2,1),(5,1),(3,2),(5,2),(4,3),(5,4)\}$. Let $\widetilde x= [0.737,
    0.088,
    0.410,
    0.125,
   -1.362]^{\mathsf{T}}$, $\widetilde{z}=[0,0,0,0,1,1]^{\mathsf{T}}$, $\alpha=0.1$, and $\beta=1$. Then, the incidence matrix $A$  and the vector of measurements can be easily constructed as
\begin{gather*}
A=\left[
\begin{array}{ccccc}
1&-1&0&0&0\\
1&0&0&0&-1\\
0&1&-1&0&0\\
0&1&0&0&-1\\
0&0&1&-1&0\\
0&0&0&1&-1
\end{array}
\right]
\end{gather*}
and $b=[      0.658,
    2.105,
   -0.322,
    1.450,
   -0.094,
    1.190
]^{\mathsf{T}}.$
The resulting estimates are $\widehat{x}^{\text{wls}}=[   0.737,
    0.078,
    0.397,
    0.156,
   -1.368]^{\mathsf{T}}$ by weighted least squares, $\widehat{x}^{\text{ls}}=[ 0.803,
    0.084,
    0.222,
    0.132,
   -1.242
]^{\mathsf{T}}$  by unweighted least squares and 
 $\widehat{x}^{\text{lae}}=[0.803,
    0.144,
    0.242,
    0.112,
   -1.302]^{\mathsf{T}}$ by $\ell_1$-minimization.
We obtain that $\|\widehat{x}^{\text{wls}}-\widetilde{x}\|^2/\|\widetilde{x}\|^2=   4.89\cdot10^{-4}
$ and $\|\widehat{x}^{\text{ls}}-\widetilde{x}\|^2/\|\widetilde{x}\|^2=     
    2.09\cdot10^{-2}$ and $\|\widehat{x}^{\text{lae}}-\widetilde{x}\|^2/\|\widetilde{x}\|^2=     
    1.52\cdot10^{-2}$.
\hfill\QED
\end{es}
%
%
%

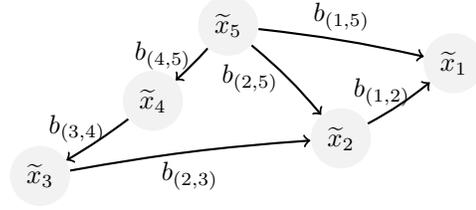
\begin{figure}\begin{center}
\begin{tikzpicture}
 [scale=0.5]
  \node [ circle,fill=gray!10 ] (n4) at (8,7)  {$\widetilde x_4$};
  \node  [circle,fill=gray!10 ] (n5) at (10,9)  {$\widetilde x_5$};
  \node  [circle,fill=gray!10 ](n1) at (16,8) {$\widetilde x_1$};
  \node [circle,fill=gray!10](n2) at (13,6)  {$\widetilde x_2$};
  \node [circle,fill=gray!10 ](n3) at (5,5)  {$\widetilde x_3$};
\draw (n4) edge[thick,<-,bend right=3] node[left, pos=.8]{$b_{(4,5)}$} (n5);
\draw (n1) edge[thick,<-,bend right=3] node[above]{$b_{(1,5)}$} (n5);
\draw (n1) edge[thick,<-,bend right=-5] node[above, pos=0.8]{$b_{(1,2)}$} (n2);
\draw (n2) edge[thick,<-,bend right=5] node[left]{$b_{(2,5)}$} (n5);
\draw (n2) edge[thick,<-,bend right=3] node[below]{$b_{(2,3)}$} (n3);
\draw (n3) edge[thick,<-,bend right=3] node[left,pos=0.8]{$b_{(3,4)}$} (n4);
\end{tikzpicture}
\caption{The network of 5 nodes considered in Example~\ref{es1}. }
\label{fig:es1}
\end{center}
\end{figure}

\section{Centralized algorithm}\label{sect:algo}
In this section, we tackle the likelihood maximization problem~\eqref{eq:first-ML} in its full generality. Since \eqref{eq:first-ML} does not admit a closed form solution, we propose an iterative algorithm that provides a solution in an iterative fashion.
Preliminarily to designing our algorithm, we convert the Maximum Likelihood problem into a minimization problem by the following result, whose proof is postponed to Appendix~\ref{sect:likelihood}.
\begin{teo}\label{prop:ML-V}The following optimization problems have the same solutions
\begin{gather}
\label{min0}\max_{\alpha, \beta}\max_{x}L(x,\alpha, \beta)\\
\label{min1}
-\min_{\alpha, \beta}\min_{x}\min_{\pi\in [0,1]^\E}V(x,\pi,\alpha,\beta)
\end{gather}
where 
\begin{align}\label{eq:loglikelihood}
& V(x,\pi,\alpha,\beta) \\
\nonumber & \quad  =  \frac{1}{2}\sum_{e\in\E}(b-Ax)^2_e\left(\frac{1-\pi_e}{\alpha^2}+\frac{\pi_e}{\beta^2}\right) \\
\nonumber & \quad \quad \: {+}\sum_{e\in\E}\left[{-}\pi_e\log\frac{p}{\beta}-(1-\pi_e)\log\frac{1-p}{\alpha}  -H(\pi_e)\right]
\end{align}
and $\map{H}{[0,1]}{\R}$ is the natural entropy function
$
H(\xi)=-\xi\log\xi-(1-\xi)\log(1-\xi).
$
\end{teo}

\smallskip
Note that, with respect to the original problem~\eqref{min0}, problem~\eqref{min1} explicitly introduces the variable $\pi\in [0,1]^\E$ which represents the estimated probabilities that the edges have large variances.
\lessnew{Actually, instead of solving problem~\eqref{min1}, we will solve a suitably modified problem, which we are going to define next. This modification marks a key difference with classical EM approaches. Namely, we shall solve}
\begin{equation}\label{min2}
\min_{\alpha,\beta}\min_{x}\min_{\pi\in \Sigma_{|\mathcal{E}|-s}} \tilde V(x,\pi,\alpha,\beta,\epsilon),
\end{equation}
where $\map{\tilde V}{\R^\V\times[0,1]^\E\times\realpositive\times\realpositive\times\realnonnegative}{\R}$ is 
\begin{align}\label{Lyapunov}
\tilde V(x,\pi,&\alpha,\beta,\epsilon)\\
\nonumber=&\frac{1}{2}\sum_{e\in\E}\left((b-Ax)^2_e+\frac{\epsilon}{|\E|}\right)\left(\frac{1-\pi_e}{\alpha^2}+\frac{\pi_e}{\beta^2}\right)\\
\nonumber&{+}\sum_{e\in\E}\left[{-}\pi_e\log\frac{p}{\beta}-(1-\pi_e)\log\frac{1-p}{\alpha}-H(\pi_e)\right].
\end{align}
\lessnew{Compared to~\eqref{min1}, the optimization problem~\eqref{min2} introduces 
\begin{itemize}
\item the positive variable $\epsilon$, which has the goal to avoid possible singularities when one of the
Gaussian components of the mixture collapses to one point;
\item the constraint set $\Sigma_{|\mathcal{E}|-s}$, which implies that at least $s\geq 1$ measurements are classified as reliable. 
\end{itemize}
As will become clear in the proofs, these modifications are instrumental to guarantee the convergence of the algorithms that we design. 
By defining function $\tilde V$, we do not intend to pose any additional assumption in our original problem statement~\eqref{eq:first-ML}. 
However, problems~\eqref{min1} and~\eqref{min2} are not equivalent: instead, Problem~\eqref{min2} should be seen as a treatable approximation of~\eqref{min1}. The mismatch between the two problems is meant to be small, since $\eps$ is bound to be small and it suffices to choose $s$ as small as $1$.
}

The following lemma summarizes the main properties of $\tilde V$ in minimization problems that only involve one variable at the time. Its proof can be obtained by differentiating $\tilde V$. 
\begin{proposition}[Partial minimizations]\label{prop:pi}
Let us define
\begin{align*}
\widehat{x}&=\widehat{x}(\pi,\alpha,\beta,\epsilon)=\argmin{x\in\R^{\V}}\tilde V(x,\pi,\alpha,\beta,\epsilon)\\
\widehat{\pi}&=\widehat{\pi}(x,\alpha,\beta,\epsilon)=\argmin{\pi\in{\Sigma_{|\mathcal{E}|-s}}}\tilde V(x,\pi,\alpha,\beta,\epsilon)\\
\widehat{\alpha}&=\widehat{\alpha}(x,\pi,\beta,\epsilon)=\argmin{\alpha>0}\tilde V(x,\pi,\alpha,\beta,\epsilon)\\
\widehat{\beta}&=\widehat{\beta}(x,\pi,\alpha,\epsilon)=\argmin{\beta>0}\tilde V(x,\pi,\alpha,\beta,\epsilon)
\end{align*}
and denote $W=\mathrm{diag}\left(\frac{1-\pi}{\alpha^2}+\frac{\pi}{\beta^2}\right)$, 
$L_{W}=A^{\top}WA$, and 
$\xi_e=f(z_e=1|x,\alpha,\beta).$
%
Then, it holds true that
\begin{align*} 
\widehat{x}=L^{\dag}_W A^\top {W} b
\end{align*}
{\begin{gather*}
\widehat{\pi}=\mathcal{P}_{s}(\xi)
\end{gather*}}
$$\widehat{\alpha}=\sqrt{\frac{\sum_{e}(1-\pi_e)|b_e-(Ax)_e|^2+\epsilon}{\|\mathds{1}-\pi\|_1}}
$$
$$\widehat{\beta}=\sqrt{\frac{\sum_{e}\pi_e|b_e-(Ax)_e|^2+\epsilon}{\|\pi\|_1}}
$$
where $\map{\mathcal{P}_s}{[0,1]^\mathcal{E}}{\Sigma_{|\mathcal{E}|-s}}$ is the projection that zeroes the $s$ smallest components of the given vector. 
\end{proposition}
From the expressions of $\widehat \alpha$ and $\widehat \beta$, we can notice that the regularization term makes them greater than zero, and consequently also $\|\pi\|_1$. 

Using the insights obtained by Proposition~\ref{prop:pi}, we propose an alternating method for the minimization of~\eqref{min2}. The resulting method is a combination of Iterative Reweighted Least Squares (IRLS) and Expectation Maximization.
The algorithm, which is detailed in Algorithm~\ref{algo:EM}, is based on the following four fundamental steps, which are iteratively repeated until convergence.

{\bf WLS solution:} Given the relative measurements $b$ and current parameters $\pi_e,\alpha, \beta$, a new estimation of the variable $x$ is obtained  
by solving the WLS problem with weights $$w_{e}=(1-\pi_{e})\alpha^{-2}+\pi_{e}\beta^{-2}, \quad \forall e \in \mathcal{E}.$$ 

{\bf Expectation:} \lessnew{The posterior distribution $\xi$ of the noise associated to the edges} is evaluated, based on the current $x,\alpha, \beta$.  

{\bf Projection:} \lessnew{The vector $\pi=\mathcal{P}_{s}(\xi)$ is the best $(|\mathcal{E}|-s)$-approximation of the posterior probability $\xi$.
Therefore, the $s$ smallest components of the posterior probabilities $\xi$ are \lessnew{set to zero to make sure that at least $s$ measurements are used} in the WLS estimation problem.}

{\bf Maximization:} Given the projected posterior probability $\pi$, we use it to re-estimate the mixture parameters  $\alpha $ and $\beta $. 

The procedure is iterated until a suitable Stopping Criterion (SC) is satisfied, e.g. a maximum number of iterations $T_{\max}$ can be fixed (${\text{SC} = }\{t\leq T_{\max}$\}) or the algorithm can be run until the estimate stops changing $({\text{SC} = }\{\|x^{(t+1)}-x^{(t)}\|/\|x^{(t)}\|<\mathsf{tol}\}$ for some $\mathsf{tol}>0$). 

 \begin{algorithm}[h!]
 \caption{LS-EM}\label{algo:EM}
 \begin{algorithmic}[1]
\REQUIRE Data: $(b,A)$.
 Parameters: $c_1,c_2>0, p\in (0,\frac12)$, $\mathsf{tol}>0$.
  \STATE Initialization: \\
  ${t  \gets 0}, \ \alpha^{(t)} \gets\alpha_0,\ \beta^{(t)} \gets\beta_0,\ $
  $ \pi^{(t)} \gets {0},\ \epsilon^{(t)} \gets 1, $ \lessnew{ $\text{SC} \gets 1 $}.
\WHILE{\lessnew{$ \text{SC} \geq \mathsf{tol} $ }}
 \STATE Computation of weights: $\forall e\in\mathcal{E}$
$$
w^{(t+1)}_e \gets\frac{1-\pi_e^{(t)}}{(\alpha^{(t)})^2}+\frac{\pi_e^{(t)}}{(\beta^{(t)})^2}
$$
\STATE WLS solution:  $W^{(t+1)} \gets \mathrm{diag}(w^{(t+1)})$ $$
 x^{(t+1)} \gets L^{\dag}_{W^{(t+1)}} A^{\top}W^{(t+1)}b,
$$
\STATE Posterior distribution evaluation: $\forall e\in\mathcal{E}$
{\small{\begin{align*}\xi_e^{(t+1)}& \gets f(z_e=1|x^{(t+1)},\alpha^{(t)},\beta^{(t)})
\end{align*}}}
\STATE Best $(|\mathcal{E}|-\lessnew{s})$-approximation: $$\pi^{(t+1)} \gets \mathcal{P}_{s}(\xi^{(t+1)})$$
\STATE Regularization parameter: $$\kappa^{(t+1)} \gets \mathrm{dim}(\mathrm{ker}(L_{W^{(t+1)}}))$$
$$\theta^{(t+1)} \gets \frac{1}{\log(t+1)}+c_1\|x^{(t+1)}-x^{(t)}\|+c_2(\kappa^{(t+1)}-1)$$ $$\epsilon^{(t+1)} \gets\min\left(\epsilon^{(t)},\theta^{(t+1)}\right)$$
\STATE Parameters estimation: $$\alpha^{(t+1)} \gets \sqrt{\frac{\epsilon^{(t)}+\sum_{e}(1-\pi_e^{(t+1)})|b_e-(Ax^{(t+1)})_e|^2}{\|\mathds{1}-\pi^{(t+1)}\|_1}}
$$
$$\beta^{(t+1)} \gets \sqrt{\frac{\epsilon^{(t)}+\sum_{e}\pi_e^{(t+1)}|b_e-(Ax^{(t+1)})_e|^2}{\|\pi^{(t+1)}\|_1}}
$$
 \STATE \lessnew{Evaluate $\text{SC} $} 
  \STATE ${ t  \gets  t + 1 }$
 \ENDWHILE
 \end{algorithmic}
 \end{algorithm}
 
{\lessnew{{Although Algorithm~\ref{algo:EM} is a modified version of classical EM, this fact is not sufficient to guarantee the convergence of the proposed method. In fact, as  observed in \cite{Bishop:2006:PRM:1162264}, \lessnew{ a generic EM algorithm is not guaranteed to converge to a limit point but only to produce a sequence of points along which the log-likelihood function does not decrease}.
Hence, an explicit convergence proof is required in all specific cases.
\lessnew{Algorithm~\ref{algo:EM} also includes a regularization sequence $\epsilon^{(t)}$, which appears in the ``Maximization'' step and is designed to be monotonic and to go to zero upon convergence of the algorithm (see Step 7).} The presence of such regularization is actually instrumental to prove the convergence to a local maximum of the log-likelihood function.

We underline that the proposed method (see Algorithm 1) can be interpreted also as an IRLS with a specific choice of the weights \cite{DBLP:journals/tsp/BaBPB14}.
In contrast to classical IRLS, LS-EM allows to perform a classification of the measurements and the weights depend on the weighted energy based on this classification and this marks its difference with IRLS where the weights associated to edge $e$ of the residual, chosen with the aim of approximating the $\ell_1$-norm of residual, turn out to depend exclusively on the residual of edge $|b_e-(Ax)_e|$. The combination of EM with IRLS has been shown to outperform classical IRLS for in terms of speed and robustness in presence of noise in sparse estimation problems \cite{DBLP:journals/tsp/RavazziM15}.
}}}

In order to state the convergence result, denote  $\zeta^{(t)}=(x^{(t)},\pi^{(t)},\alpha^{(t)},\beta^{(t)},\epsilon^{(t)})$: then Algorithm~\ref{algo:EM} can be seen as a map from $\R^\V\times[0,1]^\E\times\realpositive\times\realpositive \times \realnonnegative $ to itself that produces the sequence of iterates $\{\zeta^{(t)}\}_{t\in\integernonnegative}$.
\begin{teo}[LS-EM convergence]\label{thm:convergence}
 For any $b\in\R^{{\E}}$, the whole sequence $\zeta^{(t)}$ converges to $\zeta^\infty=(x^{\infty},\pi^{\infty},\alpha^{\infty},\beta^{\infty},\epsilon^{\infty})$ such that
\begin{subequations}\begin{gather}\label{eq:conv1}
x^{\infty}=L_{W^{\infty}}^{\dag}A^{\top}W^{{\infty}}b,\quad  W{^{\infty}}=\mathrm{diag}(w^{\infty})\\
w^{\infty}_e=\frac{1-\pi^{\infty}_e}{|\alpha^{\infty}|^2}+\frac{\pi_e^{\infty}}{|\beta^{\infty}|^2}\label{eq:conv2}\\
\pi^{\infty}_e=\mathcal{P}_{s}\Big( f(z_e=1|x^{\infty},\alpha^{\infty},\beta^{\infty}) \Big)\label{eq:conv3}\\
{\alpha^{\infty}=\sqrt{\frac{\sum_{e}(1-\pi_e^{\infty})|b_e-(Ax^{\infty})_e|^2+\epsilon^{\infty}}{\|\mathds{1} - \pi^{\infty}\|_1}}}\\ 
{\beta^{\infty}=\sqrt{\frac{\sum_{e}\pi_e^{\infty}|b_e-(Ax^{\infty})_e|^2+\epsilon^{\infty}}{\|\pi^{\infty}\|_1}}}.\label{eq:conv4}
\end{gather}
\end{subequations}
\end{teo}
The converge point $\zeta^\infty$ is a fixed point of the algorithm and a local minimum of $\tilde V(\cdot,\cdot,\cdot,\cdot,\epsilon^\infty)$. If $\epsilon^\infty=0$, then $\zeta^\infty$ locally maximizes the log-likelihood.
\smallskip

The proof of Theorem~\ref{thm:convergence} is based on observing that function $\tilde V$ in~\eqref{min2} is a Lyapunov function that is not increasing along the sequence of iterates.
Details are postponed to Appendix~\ref{sect:theorem-6}.

\section{Distributed algorithm}\label{sect:algo-distrib}
In this section, we design and study a distributed algorithm to solve problem~\eqref{eq:first-ML}, starting from the centralized one proposed in the previous section. Preliminary, let us examine steps 3--8 in Algorithm~1 in order to identify whether they are amenable to a distributed computation. Steps 3 and 5 only require information depending on edge $e$ and are thus inherently decentralized. Furthermore, we already know that the least squares problem in Step 4 can be solved by a distributed procedure. Instead, steps 6--8 involve global information and can not easily be distributed.

Based on this discussion, we propose a simple but effective variation of LS-EM algorithm, detailed in Algorithm~\ref{algo:dEM}. \lessnew{Algorithm \ref{algo:dEM} is totally distributed and can be performed by the nodes: at each iteration, every node $v\in\V$ computes $w_e^{(t)},\pi^{(t)}_e$ for every edge incident to it (see Step 3 and Step 5 in Algorithm~\ref{algo:dEM}) and the estimate of position $x_v^{(t)}$ (see Step 4 in Algorithm~\ref{algo:dEM}).}
The new algorithm is based on two design choices. The first choice is inspired by the distributed gradient dynamics~\eqref{eq:gradient_descent}: instead of fully solving a WLS problem at each iteration, we only perform {\em one step} of the corresponding gradient iteration. In the second choice, we assume $\alpha$ and $\beta$ to be known, thus removing the need for their estimation. A further advantage is that keeping $\alpha$ and $\beta$ fixed during the evolution of the algorithm avoids certain difficulties in the convergence analysis in Appendix~\ref{sect:theorem-6} and namely removes the need for regularization and projection steps. Consequently, this distributed algorithm solves the exact estimation problem~\eqref{min1}. Even though the knowledge of $\alpha$ and $\beta$ can be a restrictive assumption, we have observed that the algorithm is fairly robust to uncertainties in these values: this quality is further discussed in Remark~\ref{rem:uncertainty} below.
 \begin{algorithm}[h!]
 \caption{Distributed LS-EM}\label{algo:dEM}
 \begin{algorithmic}[1]
\REQUIRE Data: $(b,A)$. Parameters: $ p\in (0,\frac12), \tau>0$, $0<\alpha\ll\beta $, \lessnew{$ \mathsf{tol}>0$}. 
 \STATE Initialization: $t \gets 0,\ \pi^{(t)} \gets {0}, \ x^{(t)} \gets 0,\ $ \lessnew{ $ \text{SC} \gets 1$}
\WHILE{\lessnew{$ \text{SC}  \geq \mathsf{tol} $} }
 \STATE Computation of weights: $\forall e\in\mathcal{E}$
$$
w^{(t+1)}_e  \gets \frac{1-\pi_e^{(t)}}{\alpha^2}+\frac{\pi_e^{(t)}}{\beta^2}
$$

\STATE Gradient step: $ W \gets \mathrm{diag}(w^{(t+1)}) $
$$
 x^{(t+1)} \gets (I-\tau L_W)x^{(t)} +\tau A^{\top}Wb,\  
$$
\STATE Posterior distribution evaluation: $\forall e\in\mathcal{E}$
{\small{\begin{align*}\pi_e^{(t+1)} \gets f(z_e=1|x^{(t+1)},\alpha,\beta)
\end{align*}}}
\STATE \lessnew{Evaluate $\text{SC}  $}
 \STATE ${ t  \gets  t + 1 }$
 \ENDWHILE
 \end{algorithmic}
 \end{algorithm}
The convergence of Algorithm~\ref{algo:dEM} can be proved similarly to Theorem~\ref{thm:convergence}, under the condition that the parameter $\tau$ belongs to a certain range: details are postponed to Appendix~\ref{sect:proof-th2}. 
\begin{teo}[Distributed LS-EM convergence]\label{thm:convergence2}
If $\tau<\alpha/\|A\|^2$, then for any $b\in\R^{{\E}}$ the sequence $(x^{(t)},\pi^{(t)})$  generated by Algorithm~\ref{algo:dEM} converges to $(x^{\infty},\pi^{\infty})$ such that
\begin{subequations}
\begin{gather}\label{eq:conv1b}
x^{\infty}=L_{W^{\infty}}^{\dag}A^{\top}W^{\infty}b, \quad W{^{\infty}}=\mathrm{diag}(w^{\infty})\\
w^{\infty}_e=\frac{1-\pi^{\infty}_e}{\alpha^2}+\frac{\pi_e^{\infty}}{\beta^2}\label{eq:conv2b}\\
\pi^{\infty}_e=f(z_e=1|x^{\infty},\alpha,\beta).\label{eq:conv3b}
\end{gather}
\end{subequations}
\end{teo}
The limit point $(x^{\infty},\pi^{\infty})$ is a fixed point of the algorithm and a local minimum of $V$.

\section{Numerical results}\label{sect:sims}
In this section, we provide simulations illustrating the performance of the proposed algorithms: we are mainly interested in comparing them in terms of their convergence times and final estimation errors. 

\subsection{Performance analysis of proposed algorithms}\label{sect:sims-analysis}
We examine how  the performance of proposed techniques depends on the parameters of the problem, such as the parameters of the Gaussian mixture ($\alpha$, $\beta$, $p$) and the topology of the measurement graph $\G$. As a measure of performance, we consider the normalized quadratic error (NQE), defined as 
$\mathrm{NQE}=	{\|\widehat{x} - \widetilde{x}\|^{2}}/{\|\widetilde{x}\|^{2}} * 100 \ \left[\%\right]
$. 

Let us begin by describing our baseline simulation setup in details. First, \lessnew{we generate synthetic data to define the estimation problem. The number of nodes is set to $N=50$. The $N$ components of the state vector are generated randomly according to a uniform distribution in the interval $\left(0, 1\right) $: then, the mean is subtracted yielding a state vector with mean 0. The topology is generated as Erd\H{o}s-R\'enyi random graphs with edge probability $p_{\text{edge}}$ ranging from 0.1 up to 1 (i.e., an edge is created between two arbitrary nodes with probability $p_{\text{edge}}$).
In the extreme case of $ p_{\text{edge}} = 1 $, the graph generated is the complete graph where all nodes are connected to all others}. 
We fix $ \alpha = 0.05$ and $ \beta/\alpha $ in the range from 2 to 10. Also the probability of getting a bad measurement $ p $ is taken between 0 and 1/2. The noise vector is then sampled from a normal distribution using a combination of the above parameters. 

Next, we simulate the iterative algorithms. We initialize the state vector to be all zeros and also the vector $ \pi $ to be all zeros, \lessnew{meaning that all measurements are initially presumed to be good}. For the LS-EM, an initial value for $ \alpha $ and $ \beta $ is specified: $ \alpha $ is randomly chosen from the set \{0.1, 0.2, 0.3, 0.4, 0.5\} and $ \beta=2\alpha$ in order to meet the constraint $ \beta > \alpha $. They are then held fixed for the different trials. For the Distributed LS-EM, the fixed values for $ \alpha $ and $ \beta $ are taken to be the true values.
The stopping criterion $ \text{SC} $ is chosen according to a tolerance $ \mathsf{tol} = 10^{-4} $, which has been verified to be small enough to represent numerical convergence. 
\lessnew{For the $(|\mathcal{E}|-\lessnew{s})$-approximation, we choose $s=N-1$. This ``optimistic'' choice accounts to assume a number of valid measurements that could suffice to construct a spanning tree: this assumption is not imposed on our synthetic data.
Similar $\mathcal{P}_{s}$ projections have already shown useful to accelerate convergence of iterative reweighted least square methods for estimation problems with sparsity prior~\cite{DBLP:journals/tsp/RavazziM15}.}

Then, we simulate different trials whereby for each trial the vector $ \widetilde{z} $ is regenerated.
%
%
		\begin{figure}
			\centering
			\includegraphics[width=0.92\columnwidth]{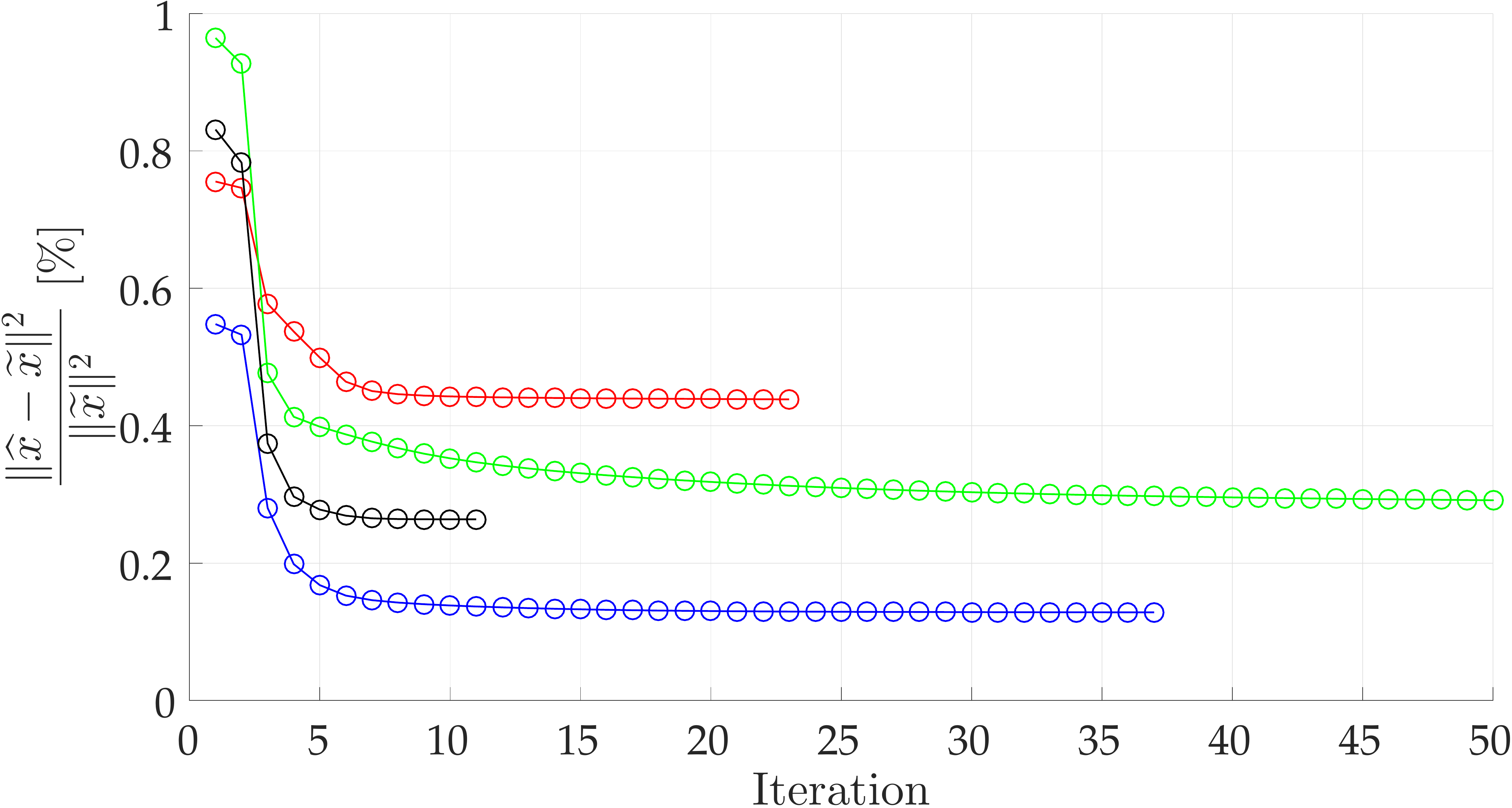}
			\caption{NQE plotted against iteration count for 4 randomly chosen trials obtained using LS-EM Algorithm (each color represents the result of a trial); the parameter set is $N = 50, p_{\text{edge}} = 0.3, p = 0.1, \alpha = 0.05, \beta/\alpha = 5 $. Note that the color of each chosen trial matches its counterpart in Fig. \ref{fig:EVO_NQE_EMDB}}\label{fig:EVO_NQE_EMCT}
		\end{figure}
		\begin{figure}
			\centering
			\includegraphics[width=0.92\columnwidth]{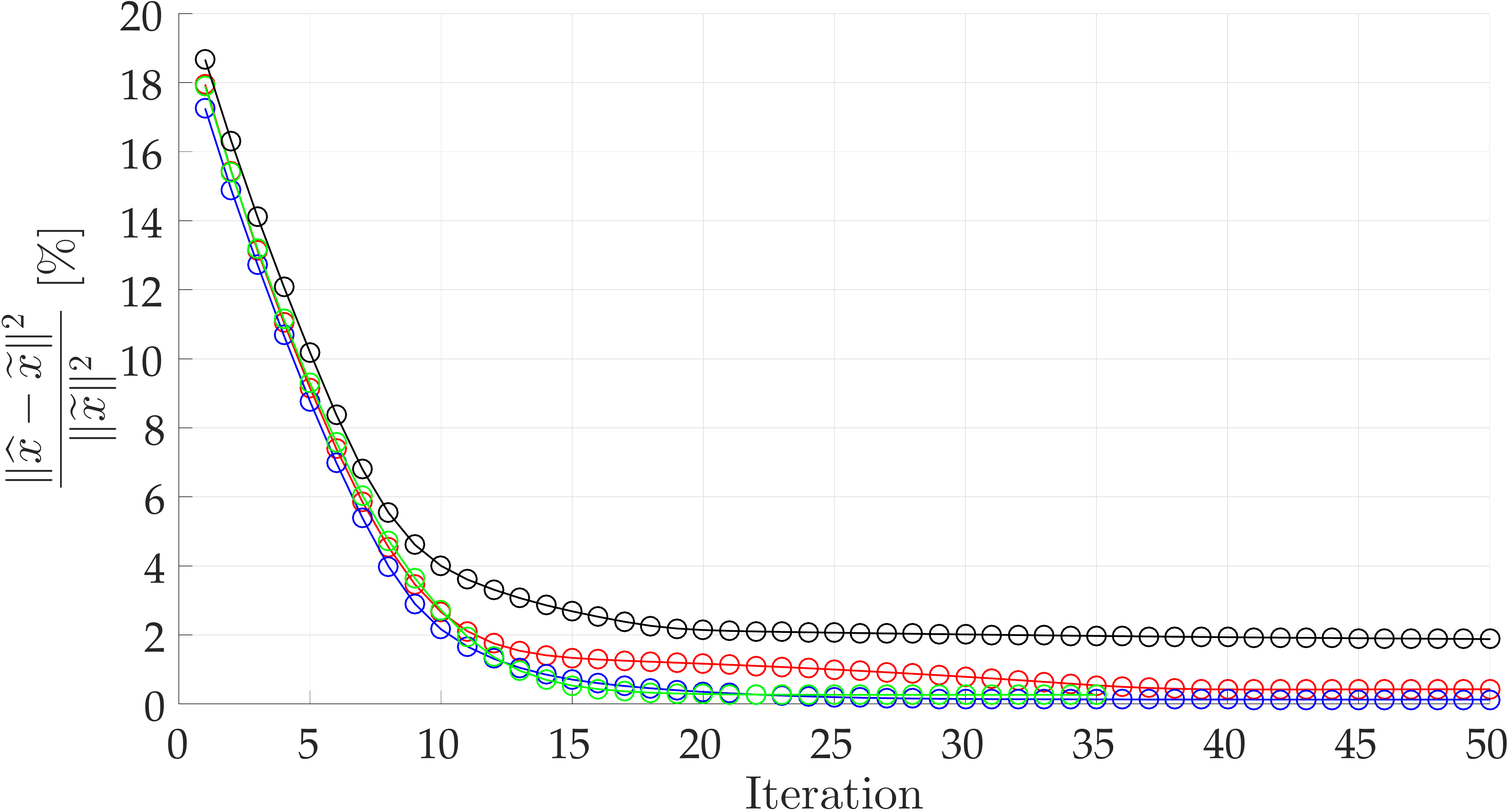}
			\caption{NQE plotted against iteration count for 4 randomly chosen trials obtained using Distributed LS-EM Algorithm (each color represents the result of a trial); the parameter set is $ N = 50, p_{\text{edge}} = 0.3, p = 0.1, \alpha = 0.05, \beta/\alpha = 5 $. Note that the color of each chosen trial matches its counterpart in Fig. \ref{fig:EVO_NQE_EMCT}}\label{fig:EVO_NQE_EMDB}
		\end{figure}
In order to illustrate the evolution of the algorithms, we plot the NQE against the iteration count for Algorithm~\ref{algo:EM} in Fig.~\ref{fig:EVO_NQE_EMCT} and Algorithm~\ref{algo:dEM} in Fig.~\ref{fig:EVO_NQE_EMDB}. We have chosen four trials out of a set of 250: the same trials (that is, the same random graphs and measurements) are chosen for both algorithms. 
We can observe that Algorithm~\ref{algo:EM} converges faster than Algorithm~\ref{algo:dEM}, but the two algorithms achieve similar final errors in a majority of trials.
		\begin{figure}
			\centering
			\includegraphics[width=0.92\columnwidth]{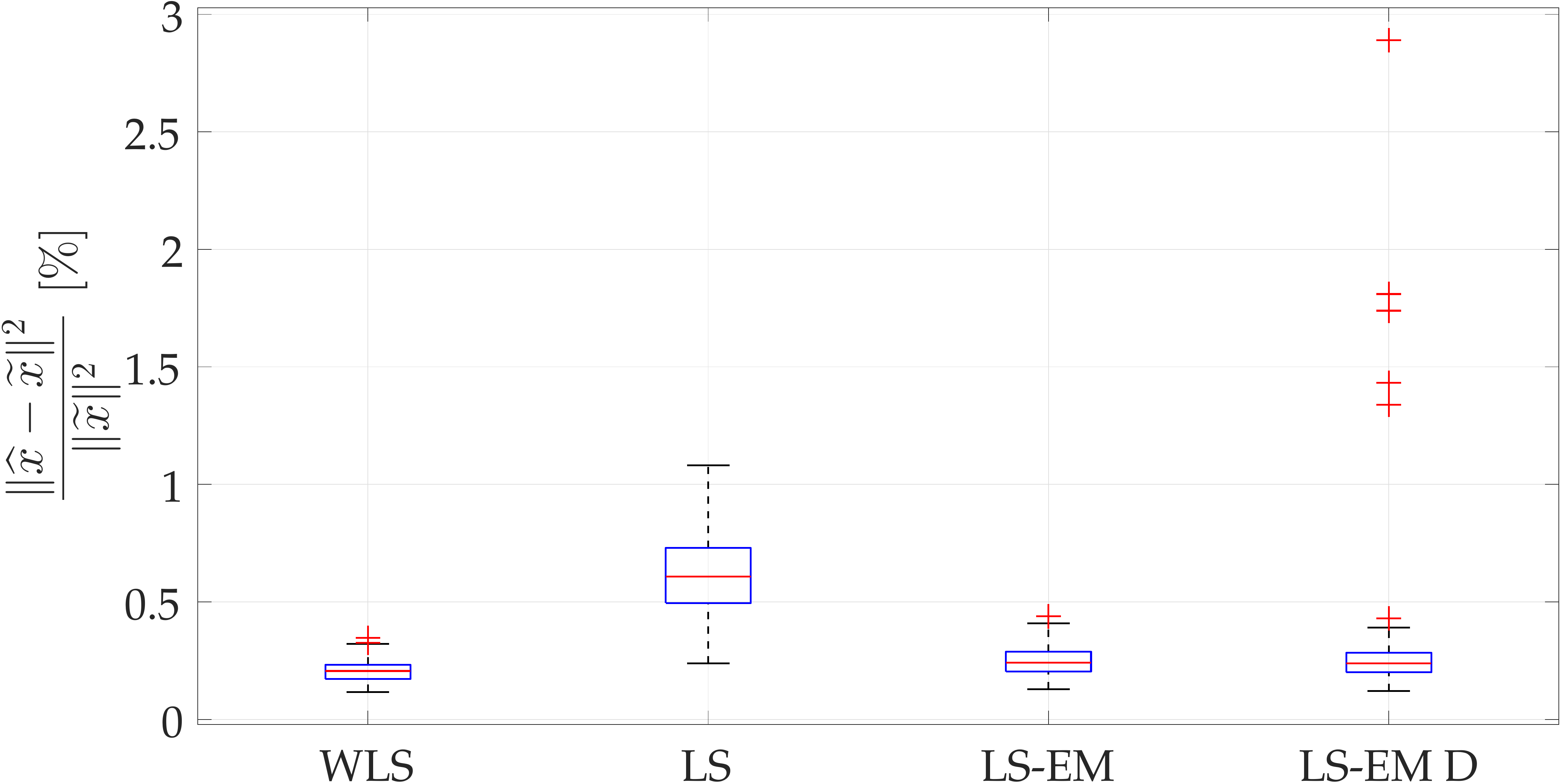}
			\caption{Boxplot showing NQE for the different approaches; the parameter set is $ N = 50, p_{\text{edge}} = 0.3, p = 0.1, \alpha = 0.05, \beta/\alpha = 5 $. + are the NQE considered as outliers by the \texttt{boxplot} command;
			}\label{fig:EVO_NQE_BOX}
			\end{figure}
			
The comparison between Algorithms~\ref{algo:EM} and Algorithm~\ref{algo:dEM} is further explored in Fig.~\ref{fig:EVO_NQE_BOX}, where the  distributions of the final NQEs for all the 250 trials mentioned above are summarized via the \texttt{boxplot} command in \texttt{MATLAB} with the default settings, showing the 25th (lower edge), 50th or median (central mark) and 75th (upper edge) percentiles. In order to make the comparison more complete and provide benchmarks, we also include the weighted least squares (WLS) as per~\eqref{eq:wls} \lessnew{and the ``naive'' unweighted least squares estimator (LS) $ \widehat{x}^{\text{ls}} $, in which we assume all measurements to be good. As expected, WLS outperforms all other estimators, thanks to using a-priori information on the noise parameters $\alpha,\beta$ and complete knowledge of the type of measurements. Instead, the naive LS has the worst performance.}

We can observe that all our approaches have a median that is clearly lower than the median of the LS approach. Actually, the bulks of the error distributions are very similar to the WLS benchmark, except for few trials of the Distributed LS-EM that perform more poorly. \lessnew{A careful inspection of these few trials shows that these large errors are due to incorrect classification of the type of a small number of edges. This phenomenon is not observed in LS-EM, possibly thanks to the fact that in a centralized algorithm the information of all nodes is used at each iteration.} 

		\begin{figure}
			\centering
			\includegraphics[width=0.92\columnwidth]{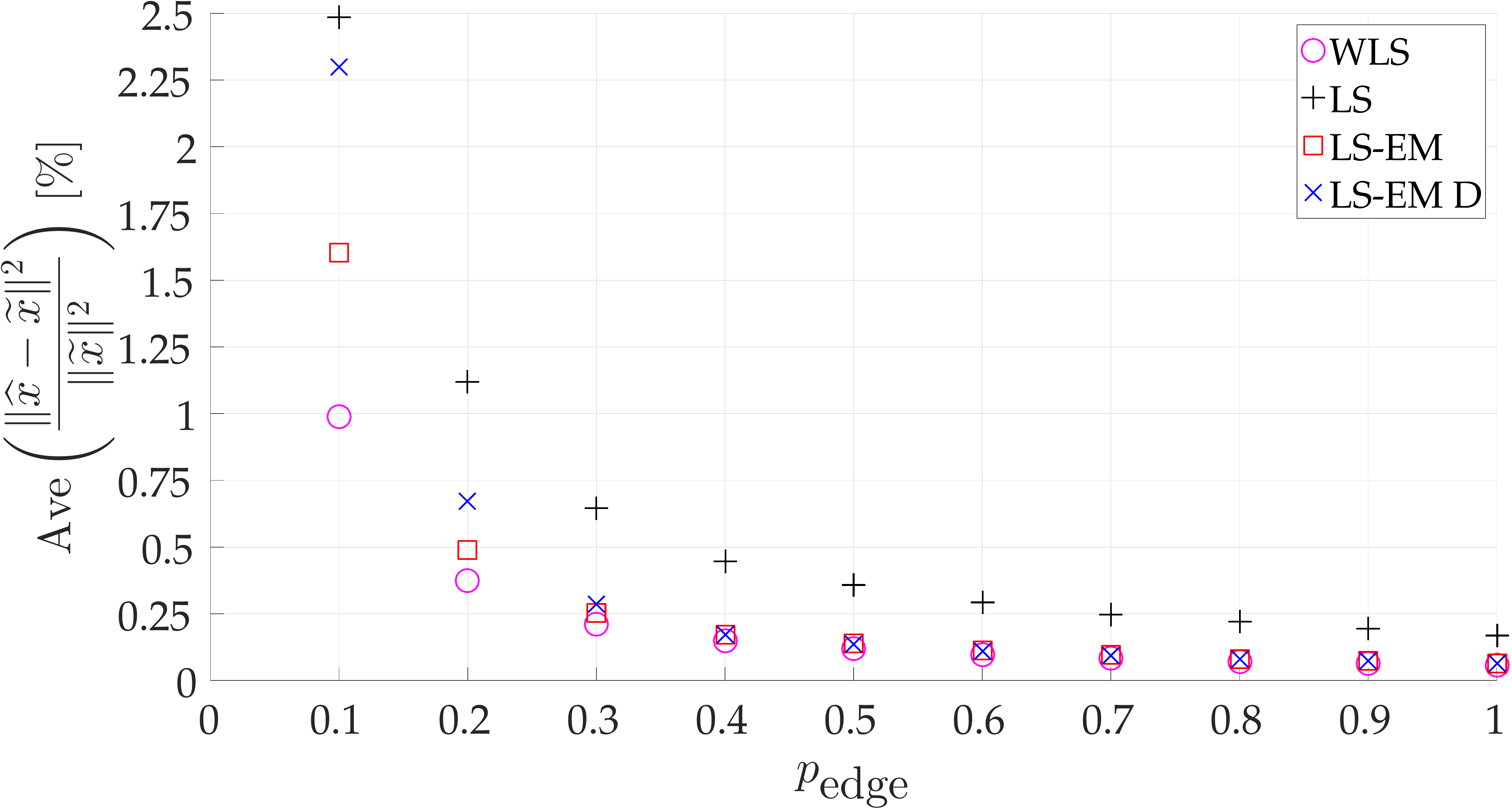}
		\caption{Mean NQE with respect to $ p_{\text{edge}} $; 
		the parameter set is $  N = 50, p = 0.1, \alpha = 0.05, \beta/\alpha = 5 $ and the number of trials is $ 1000 $. {\color{magenta}$\circ$} = WLS, {\color{black} +} = LS, {\color{red} $\Box$} = LS-EM, { $\times$} = Distributed LS-EM.}
			\label{fig:compare_pedge}
		\end{figure}

		\begin{figure}
			\centering
			\includegraphics[width=0.92\columnwidth]{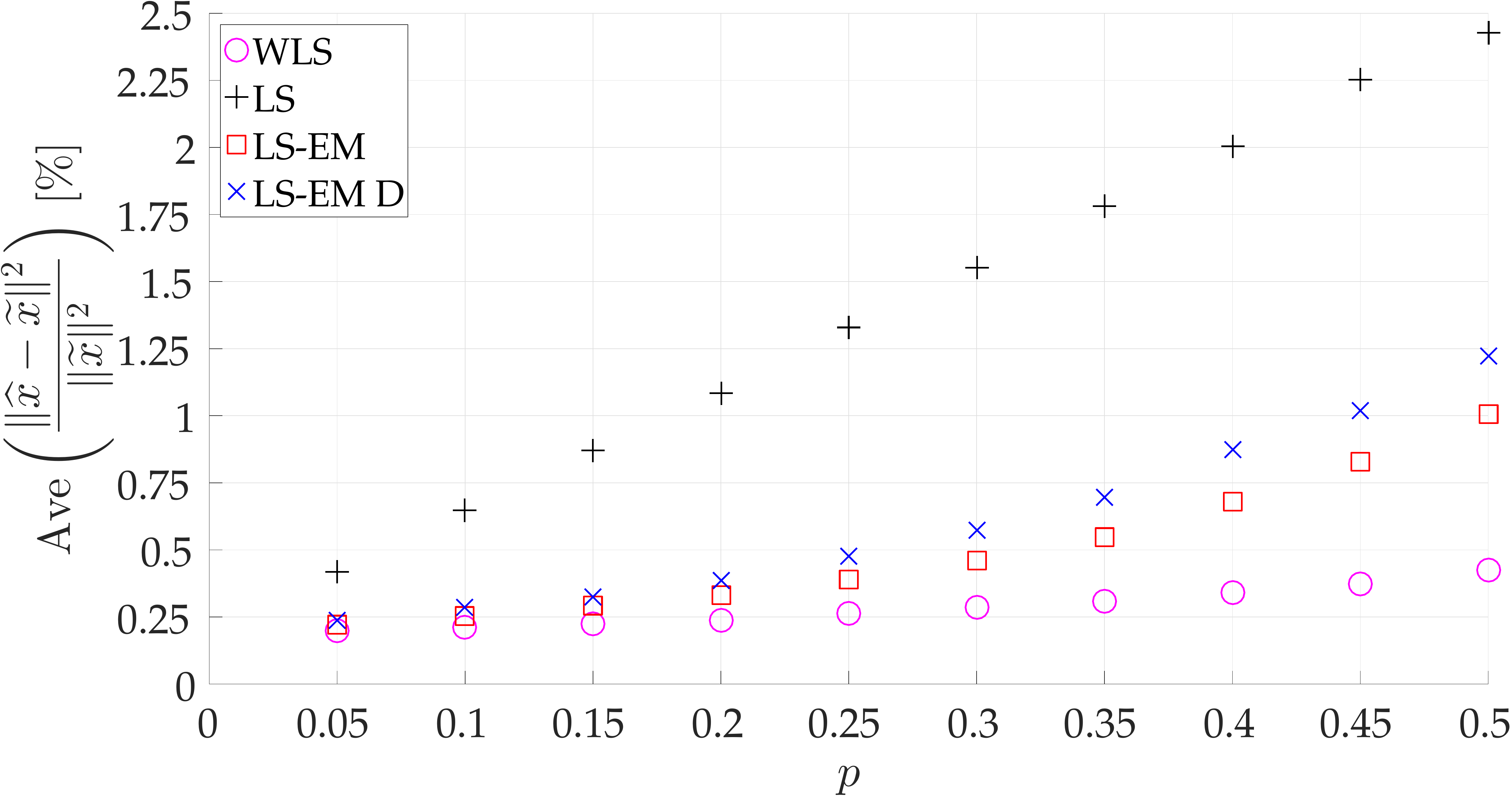}
			\caption{Mean NQE with respect to $ p $; the parameter set is $  N = 50, p_{\text{edge}} = 0.3, \alpha = 0.05, \beta/\alpha = 5 $ and the number of trials is $ 1000 $. {\color{magenta}$\circ$} = WLS, {\color{black} +} = LS, {\color{red} $\Box$} = LS-EM, $\times$ = Distributed LS-EM.}
			\label{fig:compare_p}
		\end{figure}

		\begin{figure}
			\centering
			\includegraphics[width=0.92\columnwidth]{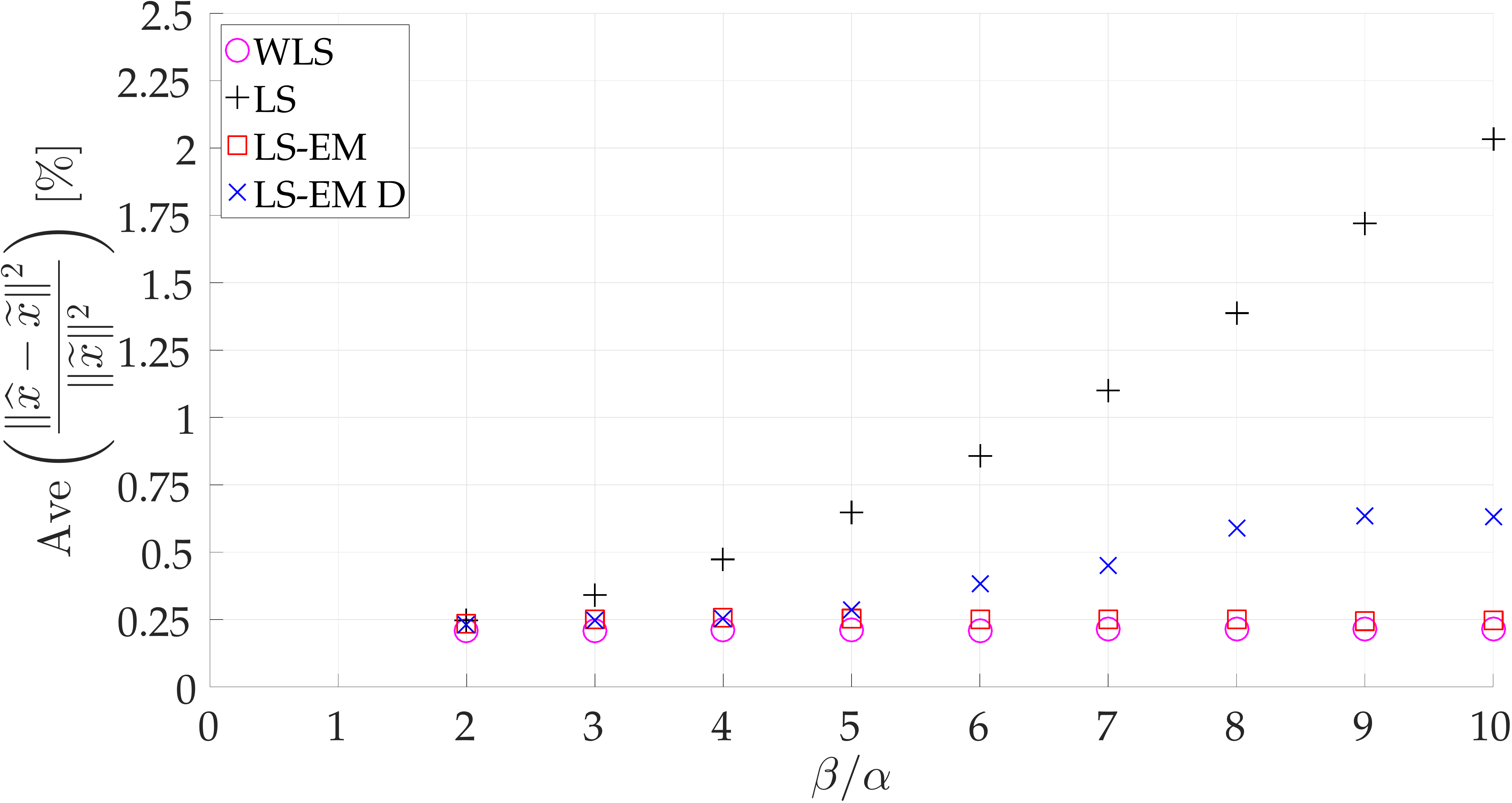}
			\caption{Mean NQE with respect to $ \beta $; the parameter set is $ N = 50, p_{\text{edge}} = 0.3, p = 0.1, \alpha = 0.05$ and the number of trials is $ 1000 $. {\color{magenta}$\circ$} = WLS, {\color{black} +} = LS, {\color{red} $\Box$} = LS-EM, { $\times$} = Distributed LS-EM.}
			\label{fig:compare_ratio}
		\end{figure}
We also performed a parameter study in order to quantify the behavior of the mean NQE with respect to $ p_{\text{edge}} $, $ p $ and $ \beta $. 
In Fig.~\ref{fig:compare_pedge}, we can observe that the mean NQE decreases with increasing $ p_{\text{edge}} $: as the graph becomes more connected, there are more measurements available to estimate the state variables. From the same figure, we can also observe that starting from $ p_{\text{edge}}=0.4$, the performance of the Distributed LS-EM is similar to that of LS-EM. In Fig.~\ref{fig:compare_p}, the mean NQE increases with increasing $ p $: this is due to the presence of more bad measurements. A similar reasoning explains the increase of NQE for increasing ratios $ \beta/\alpha $ in Fig.~\ref{fig:compare_ratio}. These dependencies on the parameters are consistent with intuition. From these three figures, it is clear that the Distributed LS-EM has larger average error than the centralised LS-EM (which has in turn a larger error than the WLS estimate). However, we should recall that the mean error of the Distributed LS-EM is driven up by the aforementioned sporadic errors: a comparison of median values would show a smaller gap from the centralized approach.

\begin{remark}[{Parameter uncertainty in Distributed LS-EM}]\label{rem:uncertainty}
Crucially, in Algorithm~\ref{algo:dEM} the values for $ \alpha $ and $ \beta $ are assumed to be known a priori. In practice one would usually not be able to have this information: hence we want to explore the sensitivity of the algorithm to incorrect choices of $\alpha$ and $\beta$.
In Fig.~\ref{fig:ROB_NQE_RC}, we assume not to know the actual values of $ \alpha $ and $ \beta $, but only their ratio $ \beta/\alpha $: we can observe that choosing $ \alpha$ and $ \beta $ to be smaller than their actual value yields a higher median, while larger values than the real one yield similar results. In Fig.~\ref{fig:ROB_NQE_RV}, we assume to know $ \alpha $ but not $ \beta$: we can observe that an incorrect and too large value of $\beta$ increases the presence of large errors, even though the bulk of the error distribution remains similar. Overall, we conclude that the algorithm is fairly robust to moderate uncertainties in the knowledge of the parameters. 
\end{remark}
		\begin{figure}[]
			\centering
			\includegraphics[width=0.92\columnwidth]{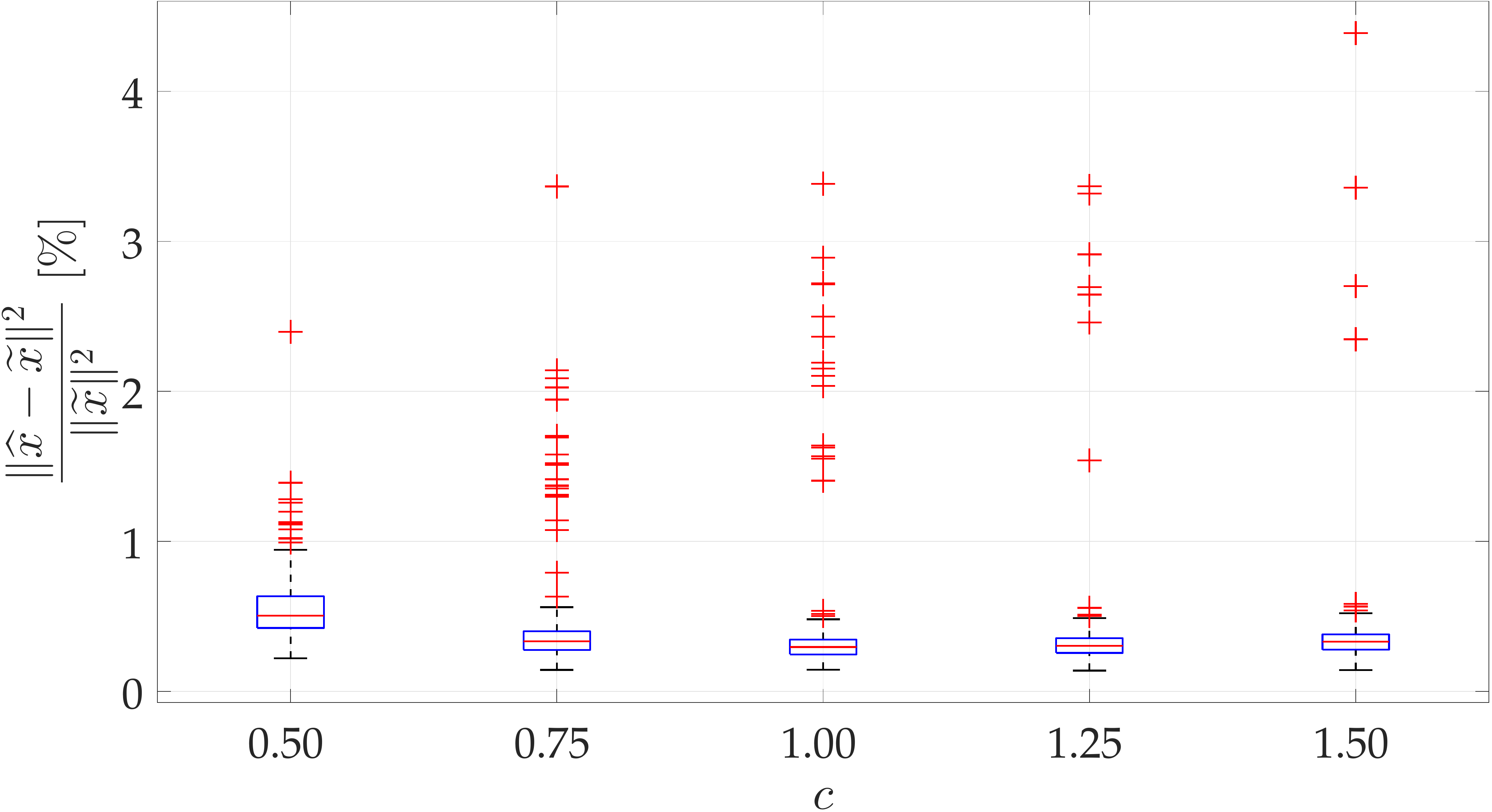}
			\caption{Boxplot of NQE for different initialization of $ \alpha $ and $ \beta $ obtained using Distributed LS-EM; we consider $ \beta/\alpha $ to be known and the actual $ \alpha $ and $ \beta $ values are set as $ c  $ times their real values, with $ c \in \left\lbrace 0.50, 0.75, 1.00, 1.25, 1.50 \right\rbrace $. The parameter set is $N = 50, p_{\text{edge}} = 0.3, p = 0.1, \alpha = 0.05, \beta/\alpha = 5$.}\label{fig:ROB_NQE_RC}
			\end{figure}

		\begin{figure}[]
			\centering
			\includegraphics[width=0.92\columnwidth]{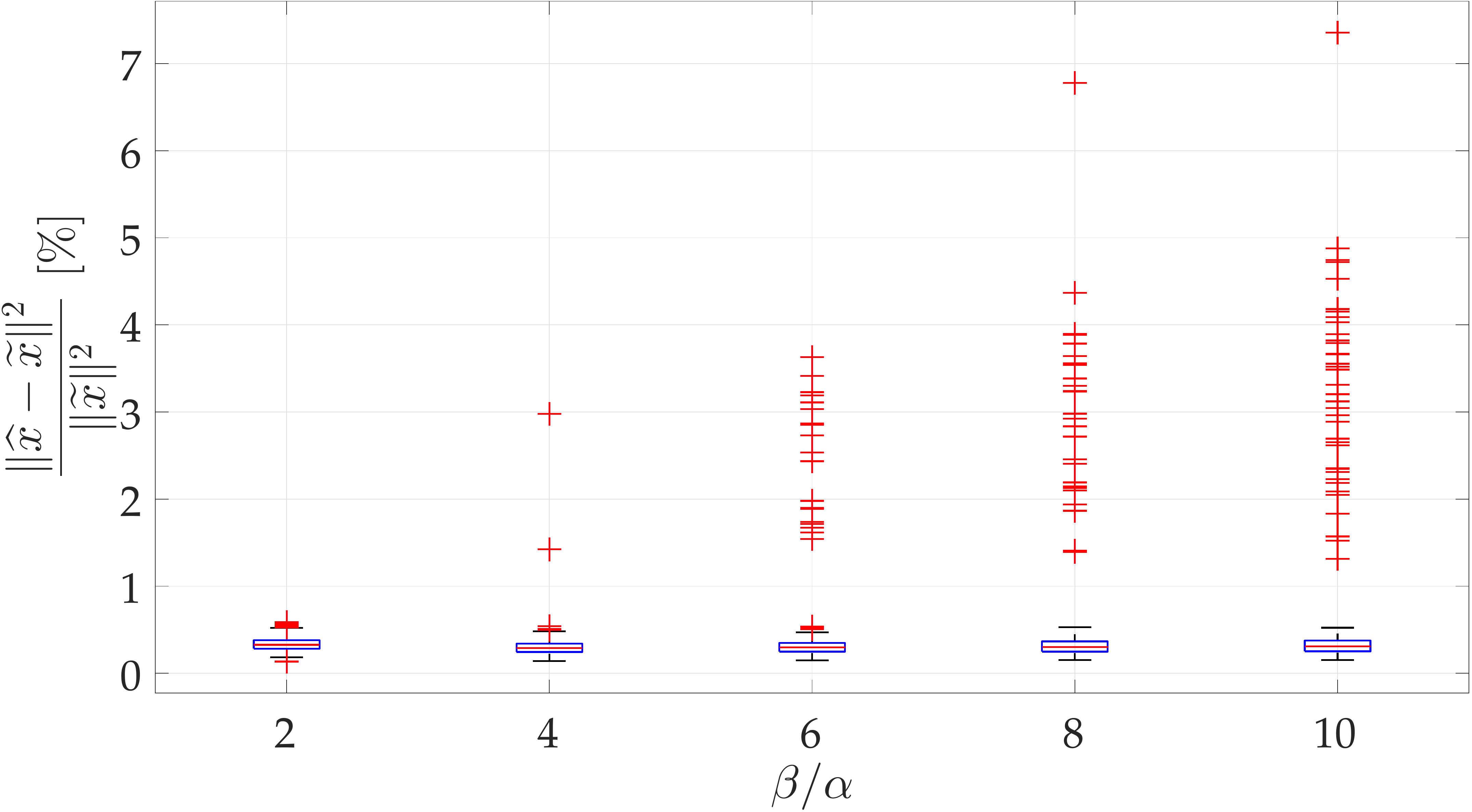}
			\caption{Boxplot of NQE for different initialization of $ \beta $ obtained using Distributed LS-EM; we consider $ \alpha $ to be known and let the actual ratio $ \beta/\alpha $ to be $ \left\lbrace 2, 4, 6, 8, 10 \right\rbrace $. The parameter set is $ N = 50, p_{\text{edge}} = 0.3, p = 0.1, \alpha = 0.05, \beta/\alpha = 5 $.}\label{fig:ROB_NQE_RV}
			\end{figure}
		
\subsection{Distributed LS-EM versus Distributed LAE}
In this section we compare Distributed LS-EM with a distributed version of LAE (see \eqref{eq:dlae} for the update). In order to perform a fair comparison we introduce a mismatch in the measurements model.
The measurements are not generated as a Gaussian mixtures and we consider an experiment coming from the geometric estimation problems in multi robot localization  \cite{LC-AC-FD:14}. More precisely, we consider the following setting.
Ground truth $\widetilde x$ of positions of $n=30$ nodes are drawn from a uniform distribution over the interval $[-1,1]$. Connections among nodes are generated according to the Erd\H{o}s-R\'enyi random
graph model, where each edge is included in the graph with
probability $p_{\mathrm{edge}}$ independently from every other edge. 
The outliers indicator vector $\widetilde{z}\sim \mathcal{B}(0.1)$ and the measurements are generated as follows
 $$
 b=A\widetilde x+(1-\widetilde{z})\alpha\eta+\widetilde{z}\gamma
 $$
 where $\eta$ is white Gaussian noise and $\gamma_e$ is distributed according to a uniform distribution over the interval $[\Delta/4, \Delta/4]$ with $\Delta$ is the size of the environment.
 In Figure 11 we show a comparison between Distributed LS-EM and Distributed LAE in terms of speed of convergence for different values of $p_{\mathrm{edge}}\in\{0.25,0.5,0,75\}$.  In order to perform a fair comparison we have fixed $p=0.2$ (and not equal to 0.1) and $\alpha =0.05$ and $\beta=5\alpha$. The figure depicts the NQE averaged over 50 experiments as a function of number of iterations. The efficiency of the proposed algorithm allows to reduce the number of iterations required to achieve a satisfactory level of accuracy. As can be noticed, Distributed LS-EM need fewer updates (about 40 itarations) than Distributed LAE (more than 300 iterations) to achieve $\mathrm{NQE} = 10^{-3}$ if $p_{\mathrm{edge}}=0.25$.
This gain reduces when the graph originated by the measurement become denser and more connected: when $p_{\mathrm{edge}}=0.5$ few iterations (about 5) are needed to guarantee the convergence of Distributed LS-EM and about 90 for Distributed LAE.
\begin{figure*}[t]
			\centering
			\includegraphics[trim=2cm 0cm 2cm 0cm, clip,width=2.1\columnwidth]{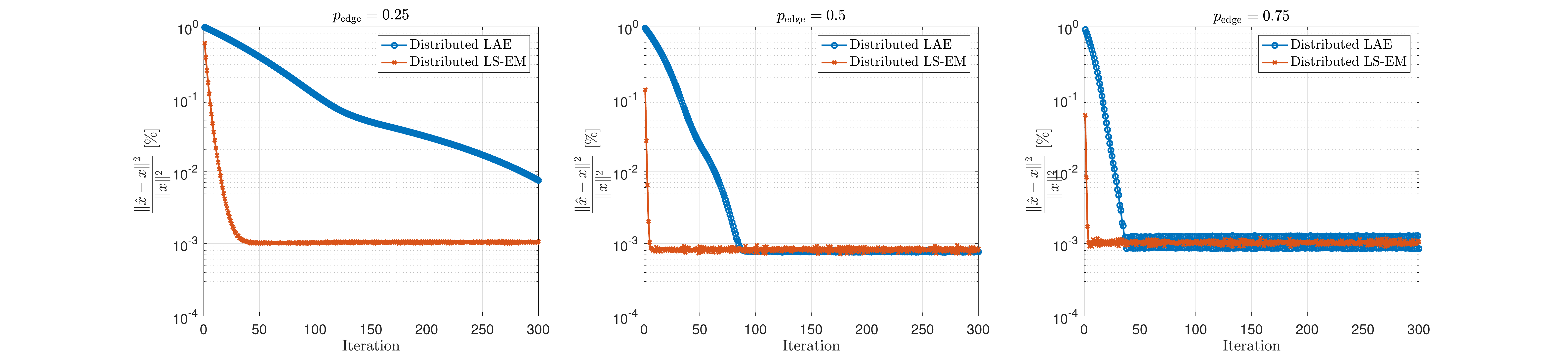}
			\caption{NQE plotted against iteration count averaged over 50 trials for Distributed LAE and Distributed LS-EM Algorithm.}\label{fig:L1vsEM1}
		\end{figure*}

\section{Concluding remarks}\label{sect:outro}
In this paper, we have studied the problem of estimation from relative measurements with heterogeneous quality. \lessnew{We have introduced a novel formulation for the problem and we have proposed two novel algorithms based on the method of Expectation-Maximization. 
One of the algorithms has the important feature of being fully distributed and thus amenable to applications where communication is limited or expensive. 
The other algorithm also distinguishes itself from standard EM approaches, due to the presence of regularization variables and of a projection step, which help dealing with the graph-dependent nature of the problem.}
Besides designing the algorithms, we have proved their convergence to a local maximum of the log-likelihood function (or to an approximation when regularization is employed). \lessnew{We note here that, as per the convergence properties, the projection step could be dispensed with at the price of a more involved proof: however, its role is not only in allowing for a proof but also in improving the performance in terms of speed, as we discussed in Section~\ref{sect:sims-analysis}.} We have also presented a number of simulations that support the good performance of the algorithms and their robustness against uncertainties in the choice of the parameters. Despite a generally good performance, the algorithms (and particularly the distributed one) may perform poorly on some instances: this could be explained by the local nature of the optimality results. 
\lessnew{It is worth mentioning that the model considered in this paper considers only one measurement per node. Our choice derives from the need to make the theoretical analysis as simple as possible. The proposed algorithms can be easily adapted to the case when each node in the network has access to multiple measurements of relative distances.
The arguments used to prove convergence can be adapted also in this case. Regarding the performance of the algorithm, we expect that repeated measurements would simply allow to lower the variance by averaging over themselves. }

Several interesting problems remain still open: we mention three of them here. 
First, one could further investigate the role of the topology of the measurement graphs in determining the performance of the algorithms: namely, some topologies could be more effective for the same number of measurements taken.
Second, one could look for distributed algorithms that need not to assume the knowledge of the mixture parameters $\alpha$ and $\beta$. 
Third, one could explore algorithms that perform a ``hard'' classification of the measurements, as opposed to the ``soft'' classification that is done in this paper, where the measurements are assigned a probability of being of type $\alpha$ or $\beta$: some preliminary results and designs with hard classification are available in~\cite{Chan2016}.

\appendix
\subsection{Properties of the likelihood}\label{sect:likelihood}

Theorem~\ref{prop:ML-V} converts the ML problem~\eqref{eq:first-ML} into a minimization problem.
Before its proof, we recall expression \eqref{eq:cond_distr} and introduce some useful notation:
{\small \begin{align}\label{eq:cond_distr2}\begin{split}
 f(b_e|x, \alpha,\beta)&=\frac{1-p}{\sqrt{2\uppi\alpha^2}}\exp{\left(-\frac{(b-Ax)_e^2}{2\alpha^2}\right)}\\
&+\frac{p}{\sqrt{2\uppi\beta^2}}\exp{\left(-\frac{(b-Ax)_e^2}{2{\beta}^2}\right)},\end{split}
\end{align}
\begin{align}\begin{split}
  f(b_e,z_e|x,\alpha,\beta)
=&\left[(1-z_e)\frac{1-p}{\sqrt{2\uppi\alpha^2}}\exp{\left(-\frac{(b-Ax)_e^2}{2\alpha^2}\right)}\right.\\
&\quad\left.+z_e\frac{p}{\sqrt{2\uppi\beta^2}}\exp{\left(-\frac{(b-Ax)_e^2}{2{\beta}^2}\right)}\right],\end{split}
\end{align}
\begin{align}\begin{split}
  f(b_e|z_e,x,\alpha,\beta)&=\left[\frac{1-z_e}{\sqrt{2\uppi\alpha^2}}\exp{\left(-\frac{(b-Ax)_e^2}{2\alpha^2}\right)}\right.\\
&\left.+\frac{z_e}{\sqrt{2\uppi\beta^2}}\exp{\left(-\frac{(b-Ax)_e^2}{2{\beta}^2}\right)}\right]\end{split}
\end{align}
} 
and
$
  f(z_e|b_e,x,\alpha,\beta)
 =\frac{f(b_e,z_e|x,\alpha,\beta)}{f(b_{e}|x,\alpha,\beta)}.
$

{\em Proof of Theorem~\ref{prop:ML-V}:} From the definition of likelihood and using \eqref{eq:cond_distr2} we have
\begin{align*}
L(x,\alpha,\beta)
&=\log\prod_{e\in\mathcal{E}} f(b_e|x,\alpha,\beta)\\
&=\sum_{e\in\mathcal{E}}\log f(b_e|x,\alpha,\beta)\\
&=\sum_{e\in\mathcal{E}}\log \sum_{z_e\in\{0,1\}}f(b_e,z_e|x,\alpha,\beta)\\
&=\sum_{e\in\mathcal{E}}\log \sum_{z_e\in\{0,1\}}\frac{f(b_e,z_e|x,\alpha,\beta)}{q(z_e)}q(z_e)
\end{align*}
for any map $\map{q}{\{0,1\}}{(0,1)}$ such that $q(0)+q(1)=1$.
From Jensen's inequality we get  
$
L(x,\alpha,\beta)\geq
\sum_{e}\sum_{z_e=0,1}q(z_e)
\log \frac{f(b_e,z_e|x,\alpha,\beta)}{q(z_e)}.$ 
Let $\pi_e=q(z_e=1)=1-q(z_e=0)$.
Therefore, we have
{\small
\begin{align}
L(x,\alpha,\beta)\geq&
\sum_{e\in\mathcal{E}}q(z_e=1)
\log \frac{f(b_e,z_e=1|x,\alpha,\beta)}{q(z_e=1)}\label{17}\\
&\quad+q(z_e=0)
\log \frac{f(b_e,z_e=0|x,\alpha,\beta)}{q(z_e=0)}\label{18}\\
&\quad=\sum_{e\in\mathcal{E}}\pi_e
\log \frac{\frac{p}{\sqrt{2\uppi\beta^2}}\exp{\left(-\frac{(b-Ax)_e^2}{2{\beta}^2}\right)}}{\pi_e}\nonumber\\
&\quad+\sum_{e\in\mathcal{E}}(1-\pi_e)
\log \frac{\frac{1-p}{\sqrt{2\uppi\alpha^2}}\exp{\left(-\frac{(b-Ax)_e^2}{2\alpha^2}\right)}}{1-\pi_e}\nonumber\\
&\quad{=}-\sum_{e\in\mathcal{E}}\left[(b-Ax)_e^2\left(\frac{\pi_e}{2{\beta}^2}+\frac{1-\pi_e}{2\alpha^2}\right)-H(\pi_e)\right]\nonumber\\
&\quad+\sum_{e\in\mathcal{E}}\left[\pi_e\log\frac{p}{\beta}+(1-\pi_e)\log\frac{1-p}{\alpha}\right]-\frac{|\E|}{2}\log(2\uppi).\nonumber
\end{align}
}
This inequality is true for all $\pi_e\in[0,1]$ and $e\in\E$ since the function on the right-hand-side can be extended by continuity in $\pi_e\in\{0,1\}$. Therefore
\begin{align}
&L(x,\alpha,\beta)\nonumber\\
&\quad\geq \max_{\pi\in[0,1]^{\E}}-\sum_{e\in\mathcal{E}}\left[(b-Ax)_e^2\left(\frac{\pi_e}{2{\beta}^2}+\frac{1-\pi_e}{2\alpha^2}\right)-H(\pi_e)\right]\nonumber\\
&\quad+\sum_{e\in\mathcal{E}}\left[\pi_e\log\frac{p}{\beta}+(1-\pi_e)\log\frac{1-p}{\alpha}\right]-\frac{|\E|}{2}\log(2\uppi)\nonumber
\end{align}
Using the definition of function $V$ in \eqref{eq:loglikelihood} we obtain
\begin{align}\label{eq:show_in}\begin{split}
L(x,\alpha,\beta)&\geq \max_{\pi\in[0,1]^{\E}}-V(x,\pi,\alpha,\beta)-\frac{|\E|}{2}\log(2\uppi)\\
&= -\min_{\pi\in[0,1]^{\E}}V(x,\pi,\alpha,\beta)-\frac{|\E|}{2}\log(2\uppi)
\end{split}
\end{align}

By differentiating $V(x,\pi,\alpha,\beta)$ with respect to $\pi_{e}$ we write the optimality condition as
\begin{align*}
\frac{\partial V}{\partial \pi_{e}}&=-\log\frac{1-\pi_{e}}{\pi_{e}}
-\log\frac{p}{{\beta}}+\frac{(b-Ax)_{e}^2}{2{\beta}^2}\\
&\quad+\log\frac{1-p}{\alpha}-{\frac{(b-Ax)_{e}^2}{2{\alpha}^2}}=0
\end{align*}
from which we obtain
\begin{align*}
\widehat{\pi}_{e}&=\frac{\frac{p}{\sqrt{2\uppi\beta^2}}\exp{\left(-\frac{(b-Ax)_{e}^2}{2{\beta}^2}\right)}}{\frac{p}{\sqrt{2\uppi\beta^2}}\exp{\left(-\frac{(b-Ax)_{e}^2}{2{\beta}^2}\right)}+\frac{1-p}{\sqrt{2\uppi\alpha^2}}\exp{\left(-\frac{(b-Ax)_{e}^2}{2{\alpha}^2}\right)}}\\
&=f(z_{e}=1|b_{e},x,\alpha,\beta).
\end{align*}
\lessnew{{Replacing  $q(z_e=1)$ and $q(z_e=0)$ with $\widehat{\pi}_e=f(z_{e}=1|b_{e},x,\alpha,\beta)$ and $(1-\widehat{\pi}_e)=f(z_{e}=0|b_{e},x,\alpha,\beta)$ in \eqref{17} and \eqref{18}, respectively, we get
\begin{align*}
&L(x,\alpha,\beta)\\
&\quad\geq
\sum_{e\in\mathcal{E}}\widehat{\pi}_e
\log \frac{f(b_e,z_e=1|x,\alpha,\beta)}{\widehat{\pi}_e}\\
&\quad+(1-\widehat{\pi}_e)
\log \frac{f(b_e,z_e=0|x,\alpha,\beta)}{1-\widehat{\pi}_e}\\
&\quad\geq
\sum_{e\in\mathcal{E}}f(z_{e}=1|b_{e},x,\alpha,\beta)
\log \frac{f(b_e,z_e=1|x,\alpha,\beta)}{f(z_{e}=1|b_{e},x,\alpha,\beta)}\\
&\quad+f(z_{e}=0|b_{e},x,\alpha,\beta)
\log \frac{f(b_e,z_e=0|x,\alpha,\beta)}{f(z_{e}=0|b_{e},x,\alpha,\beta)}\\
&\quad=\sum_{e\in\mathcal{E}}\sum_{z_e\in\{0,1\}}f(z_e|b_{e},x,\alpha,\beta)\log \frac{f(b_e,z_e|x,\alpha,\beta)}{f(z_e|b_{e},x,\alpha,\beta)}\\
&\quad=\sum_{e\in\mathcal{E}}\sum_{z_e\in\{0,1\}}f(z_e|b_{e},x,\alpha,\beta)\log f(b_e|x,\alpha,\beta)
\end{align*}
from which we conclude that the inequality in \eqref{eq:show_in} is actually an equality. }}Therefore
\begin{align*}&L(x,\alpha,\beta)=-\min_{\pi\in[0,1]^{\E}}V(x,\pi,\alpha,\beta)-\frac{|\E|}{2}\log(2\uppi).
\end{align*}
where the last expression is obtained using the definition of function $V$ in \eqref{eq:loglikelihood}.
We conclude that
$$
\max_{\alpha, \beta}\max_{x}L(x,\alpha, \beta)=\\
-\min_{\alpha, \beta}\min_{x}\min_{\pi\in [0,1]^\E}V(x,\pi,\alpha,\beta)+c
$$
with $c=
-\frac{|\E|}{2}\log(2\uppi)$.
\hfill\QEDclosed

\subsection{Proof of Theorem~\ref{thm:convergence}: convergence of Algorithm~1}\label{sect:theorem-6}
\begin{lemma}[Monotonicity]\label{lemma:Lyapunov}
The function $\tilde V$ defined in~\eqref{Lyapunov} is nonincreasing along the iterates $\zeta^{(t)}=(x^{(t)},\pi^{(t)},\alpha^{(t)},\beta^{(t)},\epsilon^{(t)}).$
\end{lemma}
\begin{proof}
Repeatedly applying Proposition~\ref{prop:pi} yields
\begin{align*}
\tilde V(\zeta^{(t+1)})&=\tilde V(x^{(t+1)},\pi^{(t+1)},\alpha^{(t+1)},\beta^{(t+1)},{\epsilon^{(t+1)}})\\
&\leq \tilde V(x^{(t+1)},\pi^{(t+1)},\alpha^{(t)},\beta^{(t)},{\epsilon^{(t)}})\\
&\leq
\tilde V(x^{(t+1)},\pi^{(t)},\alpha^{(t)},\beta^{(t)},{\epsilon^{(t)}})\\
&\leq
\tilde V(x^{(t)},\pi^{(t)},\alpha^{(t)},\beta^{(t)},{\epsilon^{(t)}})= \tilde V(\zeta^{(t)})
\end{align*}
for every time $t$, proving the result.
\end{proof}

The following lemma implies that Algorithm~1 converges numerically.
\begin{lemma}[Asymptotic regularity]\label{as_reg}
If $(x^{(t)})$ is the sequence generated by Algorithm 1, then $x^{(t+1)}-x^{(t)}\rightarrow0$ as $t\rightarrow\infty.$
\end{lemma}
\begin{proof}
From their definitions we have
\begin{align*}\alpha^{(t+1)} 
&\geq \sqrt{\epsilon^{(t)}/|\mathcal{E}|} &\qquad
\beta^{(t+1)} 
&\geq \sqrt{\epsilon^{(t)}/|\mathcal{E}|}.
\end{align*}
Then, if $\alpha^{(t)}\rightarrow0$ or $\beta^{(t)}\rightarrow0$ as $t\rightarrow\infty$, we have 
$\epsilon^{(t)}\rightarrow0$ and, consequently,
$
\|x^{(t+1)}-x^{(t)}\|_2\to0
$
and the assertion is verified.
If instead neither $\alpha$ nor $\beta$ converge to zero, then there exists a constant $K>0$ and a divergent sequence of integers $t_{\ell}$ such that  
$\min\{\alpha^{(t_{\ell})},\beta^{(t_{\ell})}\}>K$ for all $\ell\in\N$. It holds in general that 
\begin{align}\begin{split}\label{eq:lb}
&\sum_{e\in\E}{\pi_e^{(t)}}\log\beta^{(t)}+\sum_{e\in\E}{(1-\pi_e^{(t)})}\log\alpha^{(t)}-|\E|\log2\\
&\qquad\leq \tilde V(\zeta^{(t)})
\leq \tilde V(\zeta^{(1)})
\end{split}
\end{align}
where the last inequality follows from Lemma~\ref{lemma:Lyapunov}.  Then,
$
\tilde V(\zeta^{(t_{\ell})})\geq(|\E|\log K-|\E|\log2).$%

Since 
$
x^{(t+1)}=\argmin{x\in\R^{\V}}\tilde V(x,\pi^{(t)},\alpha^{(t)},\beta^{(t)},{\epsilon^{(t)}})$
we have 
\begin{align}\begin{split}\label{eq:in1}
&\tilde V(x^{(t+1)},\pi^{(t)},\alpha^{(t)},\beta^{(t)},{\epsilon^{(t)}})
\leq \tilde V(x^{(t)},\pi^{(t)},\alpha^{(t)},\beta^{(t)},{\epsilon^{(t)}})\end{split}
\end{align}
and
\begin{align}\begin{split}\label{eq:eq1}
&\nabla_{x}[\tilde V(x,\pi^{(t)},\alpha^{(t)},\beta^{(t)},{\epsilon^{(t)}})](x^{(t+1)})\\
&\qquad=A^{\top}W^{(t)}A x^{(t+1)}-A^{\top}W^{(t)}b\\
&\qquad=L_{W^{(t)}} x^{(t+1)}-A^{\top}W^{(t)}b=0\end{split}
\end{align}
where
$
W^{(t)}=\mathrm{diag}\left(\frac{1-\pi^{(t)}}{(\alpha^{(t)})^2}+\frac{\pi^{(t)}}{(\beta^{(t)})^2}\right)$ and $ 
L_{W^{(t)}}=A^{\top}W^{(t)}A.$ 
From~\eqref{eq:in1} we then have
\begin{align*}
\tilde V&(x^{(t)},\pi^{(t)},\alpha^{(t)},\beta^{(t)},{\epsilon^{(t)}})\\
&-\tilde V(x^{(t+1)},\pi^{(t+1)},\alpha^{(t+1)},\beta^{(t+1)},{\epsilon^{(t+1)}})\\
 \geq &\tilde V(x^{(t)},\pi^{(t)},\alpha^{(t)},\beta^{(t)},{\epsilon^{(t)}})\\
&-\tilde V(x^{(t+1)},\pi^{(t)},\alpha^{(t)},\beta^{(t)},{\epsilon^{(t)}})\\
=&\frac{1}{2}\sum_{e\in\E}\left((b-Ax^{(t)})^2_e+\frac{\epsilon^{(t)}}{|\E|}\right)\left(\frac{1-\pi^{(t)}_e}{(\alpha^{(t)})^2}+\frac{\pi^{(t)}_e}{(\beta^{(t)})^2}\right)\\
&-\frac{1}{2}\sum_{e\in\E}\left((b-Ax^{(t+1)})^2_e+\frac{\epsilon^{(t)}}{|\E|}\right)\left(\frac{1-\pi^{(t)}_e}{(\alpha^{(t)})^2}+\frac{\pi^{(t)}_e}{(\beta^{(t)})^2}\right)\\
=&\frac{1}{2}(x^{(t)})^{\top}L_{W^{(t)}}x^{(t)}-(x^{(t)})^{\top}A^{\top} W^{(t)}b\\
&-\frac{1}{2}(x^{(t+1)})^{\top}L_{W^{(t)}}x^{(t+1)}+(x^{(t+1)})^{\top}A^{\top} W^{(t)}b\\
=&\frac{1}{2}(x^{(t)}-x^{(t+1)})^{\top}L_{W^{(t)}}(x^{(t)}+x^{(t+1)})\\
&-(x^{(t)}-x^{(t+1)})^{\top}A^{\top} W^{(t)}b\\
=&\frac{1}{2}(x^{(t)}-x^{(t+1)})^{\top}L_{W^{(t)}}x^{(t)}\\
&+(x^{(t)}-x^{(t+1)})^{\top}(\frac{1}{2}L_{W^{(t)}}x^{(t+1)}-A^{\top} W^{(t)}b)
\end{align*}
From~\eqref{eq:eq1} we get
\begin{align*}
\tilde V(\zeta^{(t)})-\tilde V(\zeta^{(t+1)})
& \geq \frac{1}{2}(x^{(t)}-x^{(t+1)})^{\top}L_{W^{(t)}}x^{(t)}\\
&-\frac{1}{2}(x^{(t)}-x^{(t+1)})^{\top}L_{W^{(t)}}x^{(t+1)}\\
& \geq \lessnew{\frac{1}{2}(x^{(t)}-x^{(t+1)})^{\top}L_{W^{(t)}}(x^{(t)}-x^{(t+1)})}\\
&\lessnew{\geq\frac{1}{2}\min_{v:\ind^\top \!v=0}\frac{v^\top (L_{W^{(t)}})v}{\|v\|^2}\|x^{(t)}-x^{(t+1)}\|^2}.
\end{align*}
The last inequality is true since $L_{W^{(t)}}$ is positive semidefinite, the multiplicity of the eigenvalue 0 is equal to 1 and $\ind^\top x^{(t)}=0$ for all $t\in \N$.
We can thus define \begin{align*}\lambda^{(t)}&:=\frac{1}{2}\min_{v:\ \ind^\top v=0}\frac{v^\top (L_{W^{(t)}})v}{\|v\|^2}
.\end{align*}
We now prove that $\exists t_0\in\N$ such that $\lambda^{(t)}\ge c>0$ for all $t\geq t_0$. In fact, suppose by contradiction that there exists ${(t_j)}$ such that 
$\lim_{j\rightarrow\infty}\lambda^{(t_j)}=0 $.
Then, there needs to exist a subsequence $t_\ell$ such that $\alpha^{(t_{\ell})}$ or $\beta^{(t_\ell)}$ diverge. If 
$\beta^{(t_\ell)}\rightarrow\infty$, then~\eqref{eq:lb} implies $\alpha^{(t_\ell)}\rightarrow0$. From Step 8 in Algorithm~1 we obtain
$\epsilon^{(t_\ell)}\rightarrow0$ and $\kappa^{(t_\ell)}\rightarrow1$. We deduce that there exists $\ell_0\in\N$ such that $\kappa^{(t_{\ell})}=1$ for all $\ell>\ell_0$, from which we get the contradiction $\lambda^{(t_\ell)}>c>0$ for all $\ell>\ell_0$. The case $\alpha^{(t_\ell)}\rightarrow\infty$ is analogous.

We now compute for $\ell\in\N$
\begin{align*}
0\leq\sum_{t=1}^{t_{\ell}-1}c\|x^{(t)}-&x^{(t+1)}\|^2
\leq\sum_{t=1}^{t_{\ell}-1}\big(\tilde V(\zeta^{(t)})-\tilde V(\zeta^{(t+1)})\big)\\
&=\tilde V(\zeta^{(1)})
-\tilde V(\zeta^{(t_{\ell})})	\\	
&\leq \tilde V(\zeta^{(1)})-(|\E|\log K-|\E|\log 2)=K'.
\end{align*}
By letting $\ell\rightarrow\infty$, we obtain that  
$\|x^{(t)}-x^{(t+1)}\|\rightarrow0.$
\end{proof}

\begin{lemma}\label{lemma:boundness}
The sequence $(x^{(t)})_{t\in\N}$ is bounded.
\end{lemma}
\begin{proof}
If  $(\alpha^{(t)})_{t\in\N}$ and  $(\beta^{(t)})_{t\in\N}$ are both upper bounded by a constant $\chi>0$, then\begin{align*}
0&\leq\|b-Ax^{(t)}\|_2^2\\
&\leq\sum_{e\in\E}{(1-\pi_e^{(t)})[\alpha^{(t)}}]^2+\sum_{e\in\E}{\pi_e^{(t)}[\beta^{(t)}}]^2\leq \chi^2 |\E|,
\end{align*}
which guarantees that $x^{(t)}$ is bounded as well. 
%
Next, we will show that if either $\alpha^{(t)}$ or $\beta^{(t)}$ were unbounded, $x(t)$ would actually be convergent and thus bounded.

\lessnew{To this purpose, we start by observing from~\eqref{Lyapunov} that}\begin{align}\begin{split}\label{eq:alert}
&\lessnew{\sum_{e\in\E}{\pi_e^{(t)}}\log\beta^{(t)}+\sum_{e\in\E}{(1-\pi_e^{(t)})}\log\alpha^{(t)}-|\E|\log2}\\
&\lessnew{\qquad\leq \tilde V(x^{(t)},\pi^{(t)},\alpha^{(t)},\beta^{(t)},{\epsilon^{(t)}})}\\
&\lessnew{\qquad\leq \tilde V(x^{(1)},\pi^{(1)},\alpha^{(1)},\beta^{(1)},{\epsilon^{(1)}}).}\end{split}
\end{align}
\lessnew{Suppose now that $\beta^{(t)}$ is not upper bounded. Then, there exists a subsequence $(t_\ell)_{\ell\in\N}$ such that $\lim_{\ell\rightarrow\infty}\beta^{(t_{\ell})}=\infty.$
Then, inequality~\eqref{eq:alert} implies that there are two cases: either we have $\pi^{(t_{\ell})}\rightarrow0$ for $\ell\rightarrow\infty$, or $\alpha^{(t_{\ell})}\rightarrow0$. In the former case, we have $x^{(t_{\ell})}\to L^{\dag}A^{\top}b$ (where $L$ is the limit of $L_{W^{(t_\ell)}}$) implying that $x^{(t)}$ is bounded by asymptotic regularity. In the latter case, 
 there exists $e\in\mathcal{E}$ such that $\pi_e^{(t_{\ell})}\neq0$, implying that $\lim_{\ell\rightarrow\infty}\alpha^{(t_{\ell})}=0$.
From Steps~7 and~8 in Algorithm~1 we get that $\epsilon^{(t_{\ell})}\to 0$ and consequently $\kappa^{(t_{\ell})}\to 1$.
Being $\kappa^{(t)}$ an integer,} there exists $\ell_0\in\N$ such that $\kappa^{(t_{\ell})}=1$ for all $\ell>\ell_0$.
\lessnew{Since 
$$
\alpha^{(t_{\ell})}=\sqrt{\frac{\sum_{e\in\E}(1-\pi_{e}^{(t_{\ell})})|b_e-(Ax^{(t_\ell)})_e|^2+\epsilon^{(t_\ell)}}{\sum_{e\in\E}(1-\pi_{e}^{(t_{\ell})})}},
$$ we have that if there exist $\epsilon>0$ and $\{t_{\ell_j}\}_{j\in\N}$  such that $|b_e-(Ax^{(t_{\ell_j})})_e|>\epsilon$ then $\pi^{(t_{\ell_j})}_e\rightarrow 1$ as $j\rightarrow\infty$. 
On the other hand, if $|b_e-(Ax^{(t_{\ell})})_e|\rightarrow0$ then $\pi^{(t_\ell)}_e\rightarrow0$.}
\lessnew{This means that
$$
\lim_{j\rightarrow\infty}\pi^{(t_{\ell_j})}_e=\begin{cases}
1&\text{if }e\in\Delta\\
0&\text{otherwise}
\end{cases},$$ where the set $\Delta$ is defined as follows
$$
\Delta=\{e\in\E:\exists\epsilon> 0\text{ and }(t_{\ell_j})_{j}\text{ s.t }|b_e-(Ax^{(t_{\ell_j})})_e|>\epsilon\}.
$$
Observe that the relative complement $\Delta^{\text{c}}=\mathcal{E}\setminus\Delta$ has cardinality not smaller than $s$: using this notation, we can deduce that}
\begin{align*}
{\lim_{j\rightarrow\infty}A^{\top}W^{(t_{\ell_j})}Ax^{(t_{\ell_j})}}&=\lim_{j\rightarrow\infty}A^{\top}W^{(t_{\ell_j})}b\\
{A^{\top}_{\Delta^{\text{c}}}A_{\Delta^{\text{c}}}\lim_{j\rightarrow\infty}x^{(t_{\ell_j})}}&=\lim_{j\rightarrow\infty}A^{\top}_{\Delta^{\text{c}}}b_{\Delta^{\text{c}}}.
\end{align*}
Since with $\kappa^{(t_{\ell_j})}=1$, the sequence $(x^{(t_{\ell_j})})_{j\in\N}$ converges 
$$
\lim_{j\rightarrow\infty}x^{(t_{\ell_j})}=(A^{\top}_{\Delta^{\text{c}}}A_{\Delta^{\text{c}}})^{\dag}A^{\top}_{\Delta^{\text{c}}}b_{\Delta^{\text{c}}}.
$$ 
and so does $x^{(t)}$ by asymptotic regularity. 

Similarly, the case of $\alpha^{(t)}$ unbounded leads to two cases: either $\pi^{(t_{\ell})}\rightarrow1$ or $\beta\to 0$. The former case is actually forbidden by the presence of at least $s$ components equal to zero. The latter case is treated in analogous way as the case $\alpha\to 0$ above:  we omit its detailed discussion.
\end{proof}

\begin{lemma}\label{lemma:fixed_points}
Any accumulation point of $\zeta^{(t)}$ is a fixed point of
Algorithm~1 and satisfies equalities~\eqref{eq:conv1}-\eqref{eq:conv4}.
\end{lemma}

\begin{proof}
If $({x}^{\sharp},{\pi}^{\sharp},{\alpha}^{\sharp},{\beta}^{\sharp},\epsilon^{\sharp})$
is an accumulation point of the sequence $({x}^{(t)},{\pi}^{(t)},{\alpha}^{(t)},{\beta}^{(t)},{\epsilon}^{(t)})_{t\in\N}$, then there exists a subsequence
$({x}^{(t_{\ell})},{\pi}^{(t_{\ell})},{\alpha}^{(t_{\ell})},{\beta}^{(t_{\ell})},{\epsilon}^{(t_{\ell})})_{\ell\in\N}$ that converges to $({x}^{\sharp},{\pi}^{\sharp},{\alpha}^{\sharp},{\beta}^{\sharp},\epsilon^{\sharp})$ as $\ell\rightarrow\infty.$ 
We now show~\eqref{eq:conv3}, since the other conditions are immediate by continuity.
In order to verify~\eqref{eq:conv3}, we need to prove that for all $i\in\mathrm{supp}(\pi^{\sharp})$
\begin{equation}\label{eq:FP0}
\pi^{\sharp}_i=\frac{{\exp(-\frac{|b_i-(Ax^{\sharp})_{i}|^2}{2|\beta^{\sharp}|^2})}\frac{p}{\beta^{\sharp}}}{\frac{1-p}{\alpha^{\sharp}}{\exp(-\frac{|b_i-(Ax^{\sharp})_{i}|^2}{2|\alpha^{\sharp}|^2})}+\frac{p}{\beta^{\sharp}}{\exp(-\frac{|b_i-(Ax^{\sharp})_{i}|^2}{2|\beta^{\sharp}|^2})}}
\end{equation}
and for any $i \notin\mathrm{supp}(\pi^{\sharp})$ and $j\in\mathrm{supp}(\pi^{\sharp})$  
{{\begin{align}\label{eq:FP1}
&\frac{{\exp(-\frac{|b_i-(Ax^{\sharp})_{i}|^2}{2|\beta^{\sharp}|^2})}\frac{p}{{\beta^{\sharp}}}}{\frac{1-p}{{\alpha^{\sharp}}}{\exp(-\frac{|b_i-(Ax^{\sharp})_{i}|^2}{2|\alpha^{\sharp}|^2})}+\frac{p}{{\beta^{\sharp}}}{\exp(-\frac{|b_i-(Ax^{\sharp})_{i}|^2}{2|\beta^{\sharp}|^2})}}\\
&\qquad\leq\frac{{\exp(-\frac{|b_j-(Ax^{\sharp})_{j}|^2}{2|\beta^{\sharp}|^2})}\frac{p}{\beta^{\sharp}}}{\frac{1-p}{{\alpha^{\sharp}}}{\exp(-\frac{|b_j-(Ax^{\sharp})_{j}|^2}{2|{\alpha^{\sharp}|^2}})}+\frac{p}{\beta^{\sharp}}{\exp(-\frac{|b_j-(Ax^{\sharp})_{j}|^2}{|{\beta^{\sharp}|^2}})}}\nonumber.
\end{align}{}}}

Since $\lim_{\ell\rightarrow\infty}{\pi}^{(t_{\ell})}=\pi^{\sharp}$, then there exists $\ell_0$ such that, $\forall \ell>\ell_0$ and  $\forall i\in\mathrm{supp}(\pi^{\sharp})$, $\pi_i^{(t_\ell)}\neq 0$ and
$\pi_i^{(t_\ell)}=\xi_i^{(t_\ell)}\to\pi_{i}^{\sharp},$%
%
so that \eqref{eq:FP0} is verified.

If $i\notin \mathrm{supp}(\pi^{\sharp})$, then we have to distinguish the following two cases: either (a) $\pi_i^{(t_\ell)}$ is zero eventually or (b) $\pi_i^{(t_\ell)}$ converges to zero asymptotically.
In case (a), there exists ${\ell_0}\in\N$ such that $\forall\ell>\ell_0$, $\pi_i^{(t_{\ell})}=0$, from which 
$
\xi_i^{(t_{\ell})}<\xi_j^{(t_{\ell})},\  \forall j\in \mathrm{supp}(\pi^{\sharp})
$ and~\eqref{eq:FP1} is satisfied.
In case (b), there exists a strictly positive sub-sequence $({\ell_q})_{q\in\N}$ such that $\pi_i^{(t_{\ell_q})}=\xi_i^{(t_{\ell_q})}\rightarrow \pi_i^{\sharp}=0$. Since at the same time $\xi_j^{(t_{\ell_q})}$ converges to $\pi_i^\sharp>0$ for all $ j\in \mathrm{supp}(\pi^{\sharp})$, there exists $q_0\in\N$ such that $\forall q>q_0$ we have
$
\xi_i^{(t_{\ell_q})}<\xi_j^{(t_{\ell_q})},
$ and, by letting $q\rightarrow\infty$,
\begin{align*}\begin{split}
&\frac{{\exp(-\frac{|b_i-(Ax^{\sharp})_{i}|^2}{2|\beta^{\sharp}|^2})}\frac{p}{{\beta^{\sharp}}}}{\frac{1-p}{{\alpha^{\sharp}}}{\exp(-\frac{|b_i-(Ax^{\sharp})_{i}|^2}{2|\alpha^{\sharp}|^2})}+\frac{p}{{\beta^{\sharp}}}{\exp(-\frac{|b_i-(Ax^{\sharp})_{i}|^2}{2|\beta^{\sharp}|^2})}}\\
&\qquad\leq\frac{{\exp(-\frac{|b_j-(Ax^{\sharp})_{j}|^2}{2|\beta^{\sharp}|^2})}\frac{p}{\beta^{\sharp}}}{\frac{1-p}{{\alpha^{\sharp}}}{\exp(-\frac{|b_j-(Ax^{\sharp})_{j}|^2}{2|{\alpha^{\sharp}|^2}})}+\frac{p}{\beta^{\sharp}}{\exp(-\frac{|b_j-(Ax^{\sharp})_{j}|^2}{|{\beta^{\sharp}|^2}})}}\end{split}
\end{align*}
We conclude that for all $i\notin\mathrm{supp}(\pi^{\infty})$
$$\pi^{\sharp}_{i}=0=\mathcal{P}_{s}\left(\frac{{\exp(-\frac{|b-Ax^{\sharp}|^2}{2|\beta^{\sharp}|^2})}\frac{p}{\beta^{\sharp}}}{\frac{1-p}{\alpha^{\sharp}}{\exp(-\frac{|x^{\sharp}|^2}{2|\alpha^{\sharp}|^2})}+\frac{p}{\beta^{\sharp}}{\exp(-\frac{|x^{\sharp}|^2}{2|\beta^{\sharp}|^2})}}\right)_{i}.$$

\end{proof}

Since the sequence $x^{(t)}$ is bounded (see Lemma~\ref{lemma:boundness}), there exists a subsequence $x^{(t_j)}$ such that
$  
x^{(t_j)}\to{x}^{\infty}
$. Moreover, $  
\alpha^{(t_j)}\to{\alpha}^{\infty}
$, $  
\beta^{(t_j)}\to{\beta}^{\infty}
$, $  
\pi^{(t_j)}\to{\pi}^{\infty}
$, and  $  
\eps^{(t_j)}\to{\eps}^{\infty}
$.
From Lemma~\ref{as_reg}, we get
$
\lim_{t\rightarrow\infty}x^{(t_j+1)}=\lim_{t\rightarrow\infty}x^{(t_j)}={x}^{\infty},
$
proving convergence. Finally, Lemma~\ref{lemma:fixed_points} ensures that $\zeta^{(t)}$ converges to a fixed point.

\subsection{Proof of Theorem~\ref{thm:convergence2}: convergence of Algorithm~\ref{algo:dEM}}\label{sect:proof-th2}

Let us consider the function $\map{V}{\R^{\mathcal{V}}\times[0,1]^{\mathcal{E}}}{ \R}$ defined from~\eqref{eq:loglikelihood} by fixing the variables $\alpha$ and $\beta$, together with a surrogate function $V^{\sf{S}}:\R^{\mathcal{V}}\times\R^{\mathcal{V}}\times[0,1]^{\mathcal{E}}\rightarrow \R$
\begin{equation}
V^{\sf{S}}(x,z,\pi)
=V(x,\pi)+\frac{1}{2\tau}(x-z)^{\top}(I-\tau L_W)(x-z),
\label{surro}
\end{equation}
where $$W=\mathrm{diag}\left(\frac{1-\pi}{\alpha^2}+\frac{\pi}{\beta^2}\right)$$
We let the reader verify the following two lemmas.
\begin{lemma}[Partial minimizations]\label{lemma:partial_minimizations2}
If
\begin{align*}
\widehat{x}&=\argmin{x\in\R^ {\mathcal{V}}}V^{\sf{S}}(x,z,\pi)\qquad\text{and}\qquad
\widehat{\pi}&=\argmin{\pi\in[0,1]^ {\mathcal{E}}}V(x,\pi),
\end{align*}
then
\begin{align*}
\widehat{x}&=(I-\tau L_W)z+\tau A^{\top}Wb\\
\widehat{\pi}_e&=\frac{{\frac{p}{\beta}}\e^{-\frac{|b_e-(Ax)_{e}|^2}{2\beta^2}}}{{\frac{p}{\beta}}\e^{-\frac{|b_e-(A x)_{e}|^2}{2\beta^2}}+{\frac{\left(1 - p\right)}{\alpha}}\e^{-\frac{|b_e-(A x)_{e}|^2}{2\alpha^2}}} \quad \forall\, e\in\mathcal{E}
\end{align*}
\end{lemma}

\begin{lemma}[Monotonicity]\label{lemma:Lyapunov2}
The function $V$ defined in this section is nonincreasing along the iterates $\zeta^{(t)}=(x^{(t)},\pi^{(t)}).$
\end{lemma}
%

We are now able to show that Algorithm~2 converges numerically. 

\begin{lemma}\label{as_reg2}
If $x^{(t)}$ is the sequence generated by Algorithm~\ref{algo:dEM}, then $x^{(t+1)}-x^{(t)}\rightarrow0$ as $t\rightarrow\infty.$
\end{lemma}
\begin{proof}
Define $\mu=\max_{t}\|A^{\top}W^{(t)}A\|\leq \|A\|^2/\alpha$ and $\|A\|$  is the spectral norm. Since from assumption $\tau<\alpha/\|A\|^2<\mu^{-1}$ we have
\begin{align}
0&\leq\frac{1}{2\tau}(1-\tau\mu)\|x^{(t)}-x^{(t+1)}\|^2\nonumber\\
&\leq\frac{1}{2\tau}(1-\tau\|A^{\top}W^{(t)}A\|)\|x^{(t)}-x^{(t+1)}\|^2 \label{eq:cond_nec}\\
&\leq\frac{1}{2\tau}(x^{(t)}-x^{(t+1)})^{\top}(I-\tau A^{\top}W^{(t)}A)(x^{(t)}-x^{(t+1)}).\nonumber
\end{align}
If we take the sum until $T$, then
\begin{align}
0&\leq\sum_{t=1}^{T}\frac{1}{2\tau}(x^{(t)}-x^{(t+1)})^{\top}(I-\tau A^{\top}W^{(t)}A)(x^{(t)}-x^{(t+1)})\nonumber\\
&=\sum_{t=1}^{T}\left[V^{\sf{S}}(x^{(t+1)},x^{(t)},\pi^{(t)})-V(x^{(t+1)},\pi^{(t)})\right]\label{eq:pi_t}
\end{align}
Since
$\pi^{(t+1)}=\argmin{\pi}V(x^{(t+1)},\pi)$ then we have $V(x^{(t+1)},\pi^{(t+1)})\leq V(x^{(t+1)},\pi^{(t)})$ and, combining with \eqref{eq:cond_nec} and \eqref{eq:pi_t}, we get
\begin{align*}
0&\leq\frac{1}{2\tau}(1-\tau\mu)\|x^{(t)}-x^{(t+1)}\|^2\\
&\leq\sum_{t=1}^{T}\left[V^{\mathsf{S}}(x^{(t+1)},x^{(t)},\pi^{(t)})-V(x^{(t+1)},\pi^{(t+1)})\right]\\
&\leq\sum_{t=1}^{T}\left[V^{\mathsf{S}}(x^{(t)},x^{(t)},\pi^{(t)})-V(x^{(t+1)},\pi^{(t+1)})\right]\\
&=\sum_{t=1}^{T}\left[V(x^{(t)},\pi^{(t)})-V(x^{(t+1)},\pi^{(t+1)})\right]\end{align*}
where the last inequality follows from the fact $x^{(t+1)}=\argmin{x}V^{\mathsf{S}}(x,x^{(t)},\pi^{(t)})$.
Finally, we observe that the truncated series is telescopic, from which
\begin{align*}
0&\leq\frac{1}{2\tau}(1-\tau\mu)\|x^{(t)}-x^{(t+1)}\|^2\\
&\leq V(x^{(1)},\pi^{(1)})-V(x^{(T+1)},\pi^{(T+1)})	\\	
&\leq V(x^{(1)},\pi^{(1)})
-\lambda N(\log\max\{\frac{p}\beta,\frac{1-p}\alpha\}-\log 2)=C'
\end{align*}
This last inequality holds for any $T\in\N$, then by letting $T\rightarrow\infty$, we obtain that the series is convergent, from which
we deduce that as $t\rightarrow\infty$
\begin{align*}
\frac{1}{2\tau}(x^{(t)}-x^{(t+1)})^{\top}(I-\tau A^{\top}W^{(t)}A)(x^{(t)}-x^{(t+1)})\rightarrow0
\end{align*}
 and by inequality~\eqref{eq:cond_nec} the claim is proved.
\end{proof}

\begin{lemma}\label{lemma:bound2}
The sequence $(x^{(t)})_{t\in\N}$ is bounded.
\end{lemma}
\begin{proof}
Since $\mathds{1}^{\top}x^{(t)}=0$ for all $t$, Lemma~\ref{lemma:partial_minimizations2} implies
\begin{align*}
&\|x^{(t+1)}\|\\
&=\|(I-\tau A^{\top}W^{(t)}A) \big(I-\frac{1}{N}\mathds{1}\mathds{1}^{\top}\big)x^{(t)}+\tau{A}^{\top}W^{(t)}b\|\\
&\leq\|(I-\tau \big(A^{\top}W^{(t)}A)-\frac{1}{N}\mathds{1}\mathds{1}^{\top}\big)\|\| x^{(t)}\|+\|\tau{A}^{\top}W^{(t)}b\|\\
&\leq(1-\tau \mu_2)\| x^{(t)}\|+\tau\gamma
\end{align*}
where $\mu_2=\min_{t}\|A^{\top}W^{(t)}A-\frac{1}{N}\mathds{1}\mathds{1}^{\top}\|>0$ and $\gamma=\max_{t}\|{A}^{\top}W^{(t)}b\|$ (notice that $W^{(t)}$ belongs to a finite set of matrices).
We conclude that
\begin{align*}\lim_{t\rightarrow\infty}\|x^{(t)}\|&\leq\lim_{t\rightarrow\infty} (1-\tau \mu_2)^{t}\|x^{(0)}\|+\sum_{s=0}^\infty(1-\tau \mu_2)^{s}\gamma\tau
\end{align*}
which in turn is no larger than $\frac{\gamma}{\mu_2}$.
\end{proof}

By Lemma~\ref{lemma:bound2}, the sequence $x^{(t)}$ is bounded and then there exists a subsequence $x^{(t_j)}$ such that
$  
x^{(t_j)}\to{x}^{\infty}
$, $  
\alpha^{(t_j)}\to{\alpha}^{\infty}
$, $  
\beta^{(t_j)}\to{\beta}^{\infty}
$, and $  
\pi^{(t_j)}\to{\pi}^{\infty}
$.
From Lemma~\ref{as_reg2}, we get
$
\lim_{t\rightarrow\infty}x^{(t_j+1)}=\lim_{t\rightarrow\infty}x^{(t_j)}={x}^{\infty}
.$
Since $\pi_e(x)=f(z_e=1|x,\alpha,\beta)$ is a continuous function of $x$, then also $\pi^{(t)}\rightarrow\pi^{\infty}$ and $(x^{\infty},\pi^{\infty})$ is a fixed point.

\bibliographystyle{IEEEtran}

\end{document}